\documentclass[review]{elsarticle}

\usepackage{lineno,hyperref}
\usepackage{amssymb}
\usepackage{color}
\usepackage{amsmath}
\usepackage{tikz}
\usepackage{amsfonts}
\usepackage{mathrsfs}
\usepackage{amsthm}

\usepackage{bbding}
\usepackage{mathrsfs}
\usepackage{amsmath, amsfonts, amssymb, mathrsfs, txfonts}
\usepackage{graphicx, subfigure}
\usepackage{wasysym}
\usepackage{color}
\usepackage{bm}
\usepackage{enumerate}

\usepackage{pifont}
\usepackage{txfonts}
\usepackage{amsmath}
\usepackage{graphicx}
\usepackage{amsfonts}
\usepackage{amssymb}
\usepackage{mathrsfs,psfrag,eepic,epsfig}
\usepackage{epstopdf}

\newtheorem{thm}{Theorem}[section]
 \newtheorem{cor}[thm]{Corollary}
 
 \newtheorem{prop}[thm]{Proposition}
 \theoremstyle{definition}
 \newtheorem{defn}[thm]{Definition}
 \theoremstyle{remark}
 \newtheorem{rem}[thm]{Remark}

\biboptions{numbers,sort&compress}

\modulolinenumbers[5]

\journal{Journal of \LaTeX\ Templates}

\makeatletter \@addtoreset{equation}{section}
\renewcommand{\theequation}{\arabic{section}.\arabic{equation}}

\begin{document}

\begin{frontmatter}
\title{Inverse scattering transform and soliton solutions for the focusing Kundu-Eckhaus
equation with nonvanishing boundary conditions}
\tnotetext[mytitlenote]{Project supported by the Fundamental Research Fund for the Central Universities under the grant No. 2019ZDPY07.\\
\hspace*{3ex}$^{*}$Corresponding author.\\
\hspace*{3ex}\emph{E-mail addresses}: jinjieyang@cumt.edu.cn (J.J. Yang), sftian@cumt.edu.cn and
shoufu2006@126.com (S. F. Tian),
zqli@cumt.edu.cn (Z. Q. Li)}

\author{Jin-Jie Yang, Shou-Fu Tian$^{*}$ and Zhi-Qiang Li}
\address{
School of Mathematics and Institute of Mathematical Physics, China University of Mining and Technology,\\ Xuzhou 221116, People's Republic of China\\
}

\begin{abstract}

The focusing Kundu-Eckhaus (KE) equation with non-zero boundary conditions at infinity, under two cases: simple zeros and double zeros, is investigated systematically via Riemann-Hilbert (RH) problem. We derive some new results for the equation 
including the following seven parts.
 (\uppercase\expandafter{\romannumeral1})  The analyticities and symmetries of the Jost function and the scattering matrix are analyzed with the help of the normalized Lax pair. (\uppercase\expandafter{\romannumeral2})  Based on the resulting symmetries, the corresponding discrete spectrum set and residue conditions of scattering coefficients are further obtained, which is very important to construct the formulae of solution to the original equation. (\uppercase\expandafter{\romannumeral3})  A generalized RH problem is established  by combining the analytic properties of Jost functions and modified eigenfunctions. (\uppercase\expandafter{\romannumeral4})  The RH problem is solved by the corresponding asymptotic behavior combined with the Plemelj's formulae and Cauchy operator. The expression of the solution to the focusing KE equation is given under the condition of  non-reflection. (\uppercase\expandafter{\romannumeral5})  From the reflection coefficients and discrete spectrums, the trace formula and the corresponding theta condition are given to obtain the phase difference of the initial value at the boundary.
(\uppercase\expandafter{\romannumeral6}) 
For the double zeros, there is a similar framework from the set of discrete spectral points, but the operation process is much more complicated than that of simple zeros, and new results and phenomena appear.
(\uppercase\expandafter{\romannumeral7})  Some interesting phenomena are obtained that one of the solutions is gradually to rouge waves when the spectrum points tend to singular points   by choosing appropriate parameters.
\end{abstract}


\end{frontmatter}

\linenumbers

\section{Introduction}
Studing the gauge connection between different nonlinear systems including Landau-Lifshitz equations \cite{Orfanidis-1980} and derivative nonlinear Schr\"{o}dinger
(NLS) type equations \cite{Kaup-1978,Chen-1979},
Kundu proposed the famous equation \cite{Kundu-1984}
\begin{align}\label{M1}
i\hat{q}_{t}+\frac{1}{2}\hat{q}_{xx}+|\hat{q}|^{2}\hat{q}+2\beta^{2}
|\hat{q}|^{4}\hat{q}-2i\beta(|\hat{q}|^{2})_{x}\hat{q}=0,
\end{align}
with $\beta$ is a constant, which can be reduced to nonlinear Schr\"{o}dinger equation for $\beta=0$. The equation \eqref{M1} is a completely integrable system with many good properties, such as the Lax pair \cite{Geng-1999}-\cite{Yang-2014}, optical solitons \cite{Inc-2018}, soliton collisions \cite{Yan-2018}, rouge wave solutions \cite{Bayindir-2016,Wang-2014}, soliton solutions \cite{Deng-2018,Fan-2019} and the long-time asymptotic \cite{Wang-2018}.

In the past few decades, many effective methods have been put forward for solving nonlinear integrable systems, including bilinear method \cite{Hirota-1980}, Darboux transformation \cite{Ablowitz-1991}, inverse scattering (IS) method, etc., especially Zakharov et al. \cite{Zakharov-1984} have further developed the IS theory to form a modern version of the inverse scattering, i.e., RH problem. Using this method to solve nonlinear equations \cite{RHP-1}-\cite{RHP-10} have formed a boom. It is worth noting that the above literatures about RH problems are all discussed under the  zero boundary conditions (ZBC),  while some scholars have studied the condition of non-zero boundary conditions, including nonlocal nonlinear Schr\"{o}dinger equation \cite{Ablowitz-2018}, nonlocal Sine-Gordon/Sinh-Gordon Equations \cite{SG-2018}, focusing nonlinear Schr\"{o}dinger \cite{Biondini-2014}, focusing and defocusing Hirota equations \cite{Yan-2019}, TD equation \cite{zhu-2019,zhu-2018}, modified Landau-Lifshitz equation \cite{YT-2019},  etc.  For the NZBC,
Biondini and his team made a great contribution \cite{NZBC-1}-\cite{NZBC-9}. Compared with the ZBC, the difficulties in dealing with NZBC are as follows:
(\uppercase\expandafter{\romannumeral1}). One of the difficulties is to build a suitable RH problem, but there exit multi-value functions in the process of direct scattering. In order to avoid this situation, Riemann surface (please refer to \cite{Biondini-2014}) is introduced, which is a key point. (\uppercase\expandafter{\romannumeral2}). After the Riemann surface is introduced, the problem will be transformed from the original spectral $k$-plane  to a new variable $z$-plane. At the same time the difficulty is to establish the relationship between the two planes. (\uppercase\expandafter{\romannumeral3}). According to the analytical region, the analyticity of Jost function and scattering matrix can be further judged, and their respective symmetry is also a key step.  As we know, the  IS transform  and soliton solutions for the focusing KE equation \eqref{M1} under non-zero boundary conditions have not been reported, only the Fan and his team \cite{Fan-2019} has found the single soliton solutions under simple zeros, but for the double-soliton solutions, soliton solutions under the condition of double zeros, the corresponding trace formula and theta condition have not been given. Therefore, as for the problems mentioned above, we take the focusing KE equation as the model to a detailed study and some new and interesting phenomena are given.

Our calculations are based on the Lax pair of the focusing
KE equation \eqref{M1} as follows
\begin{align}\label{M2}
\begin{split}
\Psi_{x}+ik\sigma_{3}\Psi&=(\hat{Q}-i\beta \hat{Q}^{2}\sigma_{3})\Psi,\\
\Psi_{t}+ik^{2}\sigma_{3}\Psi&=\frac{1}{2}(V+4i\beta^{2}\hat{Q}^{4}\sigma_{3}-\beta
(\hat{Q}\hat{Q}_{x}-\hat{Q}_{x}\hat{Q}))\Psi,
\end{split}
\end{align}
with \begin{align*}
\hat{Q}=\left(\begin{array}{cc}
    0  & \hat{q} \\
    -\hat{q}^{*} & 0  \\
  \end{array}\right),\qquad
  \sigma_{3}=\left(\begin{array}{cc}
    1  & 0 \\
    0 & -1  \\
  \end{array}\right),
  \end{align*}\\
and $V=2kU-2\beta U^{3}-i(U^{2}+U_{x})\sigma_{3}$, here $k$ is the spectral parameter.

Take the transformation
\begin{align}\label{M3}
\phi(x,t;k)=e^{-i\beta\int|\hat{q}|^{2}dx} \sigma_{3}\Psi(x,t;k),
\end{align}
then the Lax pair \eqref{M2} can be converted to the standard focusing NLS equation
\begin{align}\label{M4}
\begin{split}
\phi_{x}+ik\sigma_{3}\phi&=Q\phi,\\
\phi_{t}+ik^{2}\sigma_{3}\phi&=\frac{1}{2}\hat{V}\phi,
\end{split}
\end{align}
with
\begin{align*}
Q=e^{-i\beta\int|\Phi|^{2}dx}\hat{Q}=\left(\begin{array}{cc}
    0  & q \\
    -q^{*} & 0  \\
  \end{array}\right),\\
  \hat{V}=e^{-i\beta\int|\Phi|^{2}dx}V=2kQ-i(Q^{2}+Q_{x})\sigma_{3},
  \end{align*}
here $q=\hat{q}e^{-i\beta\int|\hat{q}|^{2}dx}$, which means $\hat{q}=qe^{i\beta\int|q|^{2}dx}$.
\begin{rem}\label{R1}
Introducing the following transformation, and adding the additional conditions
\begin{align}\label{M5}
q=\nu e^{i\nu_{0}^{2}t},\quad \psi=e^{i\nu_{0}^{2}\sigma_{3} t/2}\phi,
\end{align}
and
\begin{align}\label{M6}
\nu\rightarrow \nu_{\pm}, \quad x\rightarrow\infty,\quad |\nu_{\pm}|=\nu_{0}>0.
\end{align}
\end{rem}

The outline of the work is arranged as:
In section 2, the asymptotic Lax pairs are obtained based on the boundary conditions. Furthermore, the analytic and symmetric properties of Jost function and scattering matrix are inferred using the obtained asymptotic Lax pairs. The residue conditions are given via the discrete spectrum, which is needed in  inverse transformation process. In section 3, the RH problem with simple zeros and reconstructing the potential are derived by  the analyticity of the modifying the eigenfunction and the scattering coefficient. In additional, trace formulae and theta conditions are given, also some new phenomenon.
In section 4, the focusing KE equation with double zeros is investigated similarity. But there are some difference including residue conditions, trace formulae etc.
Finally, some conclusions and discussions are presented in the last section.

\section{Direct scattering problem with NZBC}
In the direct scattering process, Jost function, scattering matrix and corresponding symmetry are given by using spectral analysis. In addition, the discrete spectrum, residue condition and asymptotic analysis are also presented in the section.
\subsection{The asymptotic Lax pair of focusing KE}
Based on Remark $1.1$, the equivalent Lax pairs can be expressed in the following form
\begin{align}\label{M7}
&\left\{ \begin{aligned}
&\psi_{x}=X\psi,\quad X=-ik\sigma_{3}+\tilde{Q},\\
&\psi_{t}=T\psi, \quad T=-ik^{2}\sigma_{3}+k\tilde{Q}
-\frac{i}{2}\sigma_{3}\tilde{Q}_{x}-i\frac{1}{2}\left(
\tilde{Q}^{2}+\nu_{0}^{2}\mathbb{I}\right)\sigma_{3},
     \end{aligned}  \right.
\end{align}
with
$
\tilde{Q}=\left(\begin{array}{cc}
    0  & \nu \\
    -\nu^{*} & 0  \\
  \end{array}\right).
  $ Then considering the asymptotic Lax pair  with the non-zero boundary conditions (NZBCs) as  $x\rightarrow\pm\infty$ and taking $\nu:=q$, we have
\begin{align}\label{M8}
&\left\{ \begin{aligned}
&\psi_{x}=X_{\pm}\psi,\quad X_{\pm}=\lim_{x\rightarrow\pm\infty}X=-ik\sigma_{3}+Q_{\pm},\\
&\psi_{t}=T_{\pm}\psi, \quad
T_{\pm}=\lim_{x\rightarrow\pm\infty}T=kX_{\pm},
     \end{aligned}  \right.
\end{align}
with
\begin{align*}
Q_{\pm}=\mathop{\lim}\limits_{x\rightarrow\pm\infty}\tilde{Q}=\left(\begin{array}{cc}
    0  & q_{\pm} \\
    -q^{*}_{\pm} & 0  \\
  \end{array}\right),\qquad
  \sigma_{3}=\left(\begin{array}{cc}
    1  & 0 \\
    0 & -1  \\
  \end{array}\right),
  \end{align*}
\subsection{Riemann surface and uniformization coordinate}
It is not difficult to verify that the eigenvalue of asymptotic matrix $X_{\pm}$ are multi-valued function $\pm i\sqrt{k^{2}+q_{0}^{2}}$, and in this case, unlike the zero boundary value, to deal with this situation we need to introduce a two-sheeted Riemann surface defined by
\begin{align}\label{M9}
\lambda^{2}=k^{2}+q_{0}^{2},
\end{align}
where the two-sheeted Riemann surface completed by gluing together two copies of extended complex $k$-plane $S_{1}$ and $S_{2}$ along the cut $iq_{0}[-1,1]$ between the branch points $k=\pm iq_{0}$ obtained by the value of $\sqrt{k^{2}+q_{0}^{2}}=0$.
Introducing the local polar coordinates
\begin{align}\label{M10}
k+iq_{0}=r_{1}e^{i\theta_{1}},\quad k-iq_{0}=r_{2}e^{i\theta_{2}},\quad -\frac{\pi}{2}<\theta_{1},\theta_{2}<\frac{3\pi}{2},
\end{align}
we get a single-valued analytical function on the Riemann surface
\begin{align}\label{M11}
\lambda(k)=&\left\{\begin{aligned}
&(r_{1}r_{2})^\frac{1}{2}e^\frac{{\theta_{1}+\theta_{2}}}{2}, \quad &on\quad S_{1},\\
-&(r_{1}r_{2})^\frac{1}{2}e^\frac{{\theta_{1}+\theta_{2}}}{2}, \quad &on\quad S_{2}.
\end{aligned} \right.
\end{align}
Define the uniformization variable $z$ by the conformal mapping \cite{Faddeev-1987}
\begin{align}\label{M12}
z=k+\lambda,
\end{align}
and form \eqref{M9}, one can get two single-value function
\begin{align}\label{M13}
k(z)=\frac{1}{2}\left(z-\frac{q_{0}^{2}}{z}\right),\quad \lambda(z)=\frac{1}{2}\left(z+\frac{q_{0}^{2}}{z}\right).
\end{align}

\begin{prop}
According to conformal mapping \eqref{M12}, some propositions can be observed as follows \\
 $\blacktriangleright$ $Im k>0$ of sheet $S_{1}$ and $Im k<0$ of sheet $S_{2}$ are mapped into $Im \lambda >0$;\\
 $\blacktriangleright$ $Im k<0$ of sheet $S_{1}$ and $Im k>0$ of sheet $S_{2}$ are mapped into $Im \lambda <0$;\\
 $\blacktriangleright$ The   branch  $[-iq_{0},iq_{0}]$ of $k$-plane is mapped into the branch  $[-q_{0},q_{0}]$ of $\lambda$-plane;\\
 $\blacktriangleright$ By the Joukowsky transformation map:\\
$\bullet$ $Im \lambda>0$ into domain
\begin{align*}
D^{+}=\left\{z\in \mathbb{C}:\left(|z|^{2}-q_{0}^{2}\right)Im z>0\right\},
\end{align*}
which means the upper half of the $\lambda$-plane maps to the upper outer half of the circle of radius $q_{0}$ and the inner half of the circle of the lower half of the $z$-plane.\\
$\bullet$ \emph{ $Im \lambda<0$ into domain}
\begin{align*}
D^{-}=\left\{z\in \mathbb{C}:\left(|z|^{2}-q_{0}^{2}\right)Im z<0\right\},
\end{align*}
which stands for the lower half of the $\lambda$-plane maps to the upper inner half of the circle of radius $q_{0}$ and the outer half of the circle of the lower half of the $z$-plane.\\
$\blacktriangleright$ On the sheet $S_{1}$, $z\rightarrow\infty$ as $k\rightarrow\infty$; on the sheet $S_{2}$,  $z\rightarrow 0$ as $k\rightarrow\infty$.
\end{prop}
The results above can be summarized as the following pictures

\centerline{\begin{tikzpicture}[scale=0.5]
\filldraw[gray, line width=0.5](1,0)--(3,0) arc (-180:0:2);
\path [fill=gray] (1,0) -- (9,0) to
(9,4) -- (1,4);
\path [fill=gray] (-8,0) -- (0,0) to
(0,4) -- (-8,4);
\filldraw[white, line width=0.5](3,0)--(7,0) arc (0:180:2);
\draw[fill] (5,0)node[below]{} circle [radius=0.035];
\draw[fill] (-4,-0)node[below]{} circle [radius=0.035];
\draw[-][thick](-8,0)--(-7,0);
\draw[-][thick](-7.0,0)--(-6.0,0);
\draw[-][thick](-6,0)--(-5,0);
\draw[-][thick](-5,0)--(-4,0)node[below right]{\footnotesize$0$};
\draw[-][thick](-4,0)--(-3,0);
\draw[-][thick](-3,0)--(-2,0);
\draw[-][thick](-2,0)--(-1,0);
\draw[->](-1,0)--(0,0)[thick]node[above]{$Rek$};
\draw[-][thick](-4,-4)--(-4,-3);
\draw[-][thick](-4,-3)--(-4,-2)node[below right]{\footnotesize$-iq_{0}$};
\draw[-][thick](-4,-2)--(-4,-1);
\draw[-][thick](-4,-1)--(-4,0);
\draw[-][thick](-4,0)--(-4,1);
\draw[-][thick](-4,1)--(-4,2)node[above right]{\footnotesize$iq_{0}$};
\draw[-][thick](-4,2)--(-4,3);
\draw[->][thick](-4,3)--(-4,4)[thick]node[right]{$Imk$};
\draw[->][thick](1,0)--(2,0);
\draw[-][thick](2,0)--(3,0);
\draw[-][thick](3,0)--(4,0);
\draw[<-][thick](4,0)--(5,0)node[below right]{\footnotesize$0$};
\draw[-][thick](5,0)--(6,0);
\draw[<-][thick](6,0)--(7,0);
\draw[->][thick](7,0)--(8,0);
\draw[-](8,0)--(9,0)[thick]node[above]{$Rez$};
\draw[-][thick](5,-4)--(5,-3);
\draw[-][thick](5,-3)--(5,-2)node[below right]{\footnotesize$-iq_{0}$};
\draw[-][thick](5,-2)--(5,-1);
\draw[-][thick](5,-1)--(5,0);
\draw[-][thick](5,0)--(5,1);
\draw[-][thick](5,1)--(5,2)node[above right]{\footnotesize$iq_{0}$};
\draw[-][thick](5,2)--(5,3);
\draw[-][thick](5,3)--(5,4)[thick]node[right]{$Imz$};
\draw[fill] (7,0) circle [radius=0.055][below right][thick]node{\footnotesize$0^{+}$};
\draw[fill] (3,0) circle [radius=0.055][below right][thick]node{\footnotesize$0^{-}$};
\draw[fill] (5,2) circle [radius=0.055];
\draw[fill] (5,-2) circle [radius=0.055];
\draw[fill] (-4,2) circle [radius=0.055];
\draw[fill] (-4,-2) circle [radius=0.055];
\draw[fill] (8,2.5) circle [radius=0.035][thick]node[right]{\footnotesize$z_{n}$};
\draw[fill][red] (8,-2.5) circle [radius=0.035][thick]node[right]{\footnotesize$z_{n}^{*}$};
\draw[fill][red] (3.7,1) circle [radius=0.035][thick]node[right]{\footnotesize$-\frac{q_{0}^{2}}{z_{n}}$};
\draw[fill] (3.7,-1) circle [radius=0.035][thick]node[right]{\footnotesize$-\frac{q_{0}^{2}}{z_{n}^{*}}$};
\draw[fill] (-2.5,1.5) circle [radius=0.035][thick]node[right]{\footnotesize$z_{n}$};
\draw[fill][red] (-2.5,-1.5) circle [radius=0.035][thick]node[right]{\footnotesize$z_{n}^{*}$};
\draw[-][line width=0.8] (7,0) arc(0:360:2);
\draw[->][line width=0.8] (7,0) arc(0:220:2);
\draw[->][line width=0.8] (7,0) arc(0:330:2);
\draw[->][line width=0.8] (7,0) arc(0:-330:2);
\draw[->][line width=0.8] (7,0) arc(0:-220:2);
\end{tikzpicture}}

\noindent {\small \textbf{Figure 1.} (Color online) Left Fig: the first sheet of the Riemann surface, presenting different discrete spectral points in the $k$-plane with $Imk>0$ (gray) and $Imk<0$ (white); Right Fig: shows the discrete spectral points on the $z$-plane after introducing the transformation with $Imz>0$ (gray) and $Imz<0$ (white), and gives  the orientation of the jump contours about the RH problem, where the black spectral points represent the zeros of $s_{11}(z)$ and the red spectral points represent the zeros of $s_{22}(z)$. }
\subsection{Jost function and its analyticity}
From the eigenvalue of the spectral problem $X_{\pm}$ and $T_{\pm}$, we can find an invertible matrix $Y_{\pm}(z)$ to diagonalize the two matrices. On this basis,  the so-called Jost function solution can be obtained, that is, the simultaneous solutions of lax pair \eqref{M8}:
\begin{align}\label{M14}
\psi_{\pm}(x,t;z)=Y_{\pm}(z)e^{-i\theta(x,t;z) \sigma_{3}},
\quad x\rightarrow \pm\infty,
\end{align}
with $\theta(x,t;z)=\lambda(z)[x+k(z)t]$ and
\begin{align}\label{M15}
Y_{\pm}(z)=\left(
  \begin{array}{cc}
     1 & -\frac{iq_{\pm}}{k+\lambda} \\
     -\frac{iq_{\pm}^{*}}{k+\lambda} & 1 \\
  \end{array}
\right)=\mathbb{I}-(i/z)\sigma_{3}Q_{\pm}.
\end{align}
Let
\begin{align}\label{M16}
u_{\pm}(x,t;z)=\psi_{\pm}(x,t;z)e^{i\theta(x,t;z)\sigma_{3}}
\rightarrow Y_{\pm}(z)\quad x\rightarrow \pm\infty,
\end{align}
such that
\begin{align}\label{M17}
\begin{matrix}
u_{-}(x,t;z)=Y_{-}+\int_{-\infty}^{x}Y_{-}e^{-i\lambda(x-y)\sigma_{3}}Y_{-}^{-1}\Delta Q_{-}(y,t)u_{-}(y,t;z)e^{i\lambda(x-y)\sigma_{3}}\, dy,\\
u_{+}(x,t;z)=Y_{+}-\int_{x}^{\infty}Y_{+}e^{-i\lambda(x-y)\sigma_{3}}Y_{+}^{-1}\Delta Q_{+}(y,t)u_{+}(y,t;z)e^{i\lambda(x-y)\sigma_{3}}\, dy.
\end{matrix}
\end{align}
One can obtain the following Proposition about the analyticity of the function $u_{\pm}(z)$, as $\psi_{\pm}(z)$ from the expression \eqref{M16}.
\begin{prop}
The columns $u_{+,1}(x,t;z)$ and $u_{-,2}(x,t;z)$ are analytic in $D^{-}$ of $z$-plane, and the columns $u_{-,1}(x,t;z)$ and $u_{+,2}(x,t;z)$ are analytic in $D^{+}$ of $z$-plane, here $u_{\pm,i}(x,t;z)$ ($i=1, 2$) denote the $i$-th column of $u_{\pm}$.
\end{prop}
\begin{cor}\label{c1}
From the expression \eqref{M16}, the analyticity of $u_{\pm,i}(x,t;z)$ is the same as that of $\psi_{\pm,i}(x,t;z)$ ($i=1, 2$), that is, $\psi_{+,1}(x,t;z)$ and $\psi_{-,2}(x,t;z)$ are analytic in $D^{-}$ of $z$-plane, and the rest in the upper half plane.
\end{cor}
\begin{thm} (\textbf{Liouville's formula})
Assume that $M$ is a $n$-th order matrix and satisfies a homogeneous linear differential equation $Y'=M(x)Y$, here $Y$ is an $n$ vector. If a matrix $M$ is a solution of the differential equation,  one has $(\det Y)_{x}=trM\det Y$, furthermore $\det Y(x)=\det Y(x_{0})e^{\int_{x_{0}}^{x}trM(\zeta)d\zeta}$.

Due to $trX(x,t;z)=trT(x,t;z)=0$ in the Lax pair \eqref{M7}, from \textbf{Liouville's formula} one can infer that
\begin{align}\label{M18}
\det\psi_{\pm}(x,t;z)=\det Y_{\pm}(z)=\gamma(z),\quad z\in\Sigma.
\end{align}
with $\gamma(z)=1+q_{0}^{2}/z^{2}\neq 0$ as $z\neq\pm iq_{0}$.
For $\Sigma_{0}=\Sigma/\{\pm iq_{0}\}$, because of $\psi_{\pm}(x,t;z)$ are fundamental solutions of Lax pair \eqref{M8}, thus there exists a constant matrix $S(z)$ (it's independent of the variable $x$ and $t$) satisfies
\begin{align}\label{M19}
 \psi_{+}(x,t;z)=\psi_{-}(x,t;z)S(z),\quad z\in\Sigma_{0},
\end{align}
which implies
\begin{align}\label{M20}
\psi_{+,1}=s_{11}\psi_{-,1}+s_{21}\psi_{-,2},\quad
\psi_{+,2}=s_{12}\psi_{-,1}+s_{22}\psi_{-,2},
\end{align}
here $s_{ij}$ $(i, j=1, 2)$ are element of the matrix $S(z)$. From \eqref{M14}, one has $\det S(z)=1$. Further the analyticity of the matrix $S(z)$ is discussed.

\begin{prop}
Assume $q-q_{\pm}\in L^{1}(\mathbb{R^{\pm}})$,  the elements of $s_{11}$ and $s_{22}$ are analytic in the regin $D^{-}$ and $D^{+}$, as well as continuously to $D^{+}\cup\Sigma_{0}$ and $D^{-}\cup\Sigma_{0}$, respectively. The off-diagonal elements of $S(z)$, although not analytical, continue to $\Sigma_{0}$, for convenience, we will omit the subscript next.
\end{prop}
\begin{proof}
Resorting to \eqref{M19}, one has
\begin{align}
&s_{11}(z)=\frac{Wr\left(\psi_{+,1},\psi_{-,2}\right)}{\gamma},\label{Wr-1}\quad
s_{22}(z)=\frac{Wr\left(\psi_{-,1},\psi_{+,2}\right)}{\gamma},\\
&s_{12}(z)=\frac{Wr\left(\psi_{+,2},\psi_{-,2}\right)}{\gamma},\quad
s_{21}(z)=\frac{Wr\left(\psi_{-,1},\psi_{+,1}\right)}{\gamma},
\end{align}
with $\gamma(z)=\det E_{\pm}(z)=1+q_{0}^{2}/z^{2}$. Based on the Corollary 2.3, the Proposition can be derived.
\end{proof}
Finally,  the reflection coefficients playing an important role in the inverse problem are introduced by the following equation
\begin{align}\label{Q21}
\rho(z)=s_{21}/s_{11},\quad \tilde{\rho}(z)=s_{12}/s_{22},\quad \forall z\in\Sigma.
\end{align}
\end{thm}
\subsection{Scattering matrix and its analyticity}
To get the discrete spectrum and residue condition needed in the inverse scattering process, we need to study the symmetry of the scattering problem, which is different from the zero boundary condition. There is only one symmetry, that is, mapping $k\rightarrow k^{*}$, and the existence of Riemann surface makes the symmetry more complex under the condition of non zero boundary value. Correspondingly, there are two kinds of symmetry, which read
\begin{prop}
Two different kinds of transformations:\\
$\blacktriangleright$ Make the transformation $z\mapsto z^{*}$, which means $(k,\lambda)\mapsto(k^{*},\lambda^{*})$;\\
$\blacktriangleright$ Make the transformation $z\mapsto -q_{0}^{2}/z$, which leads
$(k,\lambda)\mapsto(k,-\lambda)$;
\end{prop}
\begin{cor}
The symmetries for the Jost function $\psi_{\pm}\in\Sigma$ are presented for $z\in\Sigma$ as follows:
\begin{subequations}
\begin{align}
\psi_{\pm}(z)&=-\sigma_{0}\psi_{\pm}^{*}(z^{*})\sigma_{0}, \label{Q22}\\
\psi_{\pm}(z)&=-\frac{i}{z}\psi_{\pm}(-\frac{q_{0}^{2}}{z})
\sigma_{3}Q_{\pm},\label{Q23}
\end{align}
\end{subequations}
satisfying
\begin{align}
\psi_{\pm,1}(z)&=\sigma_{0}\psi_{\pm,2}^{*}(z^{*}),\quad \qquad
\psi_{\pm,2}(z)=-\sigma_{0}\psi_{\pm,1}^{*}(z^{*}),\label{Q24}\\
\psi_{\pm,1}(z)&=-(\frac{iq_{\pm}^{*}}{z})\psi_{\pm,2}(-\frac{iq_{0}^{2}}{z}),\quad
\psi_{\pm,2}(z)=-(\frac{iq_{\pm}}{z})\psi_{\pm,1}(-\frac{iq_{0}^{2}}{z}).\label{Q25}
\end{align}
\end{cor}
\begin{prop}
The symmetries for the scattering matrix $S(z)$ are exhibited for $z\in\Sigma$ as follows:
\begin{align}\label{Q26}
&S^{*}(z^{*})=-\sigma_{0}S(z)\sigma_{0},\\
&S(z)=(\sigma_{3}Q_{-})^{-1}S(-q_{0}^{2}/z)\sigma_{3}Q_{+}.
\end{align}
\end{prop}
\begin{cor}
The relationship between the scattering coefficients and reflection coefficients can be derived via the above symmetries for $z\in\Sigma$
\begin{align}
s_{22}(z)=s^{*}_{11}(z^{*}),  s_{12}(z)=-s^{*}_{21}(z^{*}),\label{Q27}\\
s_{11}(z)=(q_{+}^{*}/q_{-}^{*})s_{22}(-q_{0}^{2}/z),\label{Q28}\\
s_{12}(z)=(q_{+}/q_{-}^{*})s_{22}(-q_{0}^{2}/z),\\
\rho(z)=-\tilde{\rho}^{*}(z^{*})=(q^{*}_{-}/q_{-})\tilde{\rho}(-q_{0}^{2}/z).
\end{align}
\end{cor}
\subsection{Discrete spectrum and residue condition}
The set of discrete spectrum for the scattering problem is consist of all values $k\in\mathbb{C}\setminus\Sigma$ catering eigenfunctions exist in $L^{2}(\mathbb{R})$.
 Next we discuss that these discrete spectrum are the zeros of $s_{11}(z)$ and  $s_{22}(z)$ for $z\in\mathbb{D^{-}}$ and $z\in\mathbb{D^{+}}$, respectively. Assuming that $s_{11}(z)$ has $N$ simple zeros in $\mathbb{D^{-}}\cap\{z\in\mathbb{C}: Imz<0\}$ defined by $z_{n}$, $n=1, 2, \cdots, N$, namely, $s_{11}(z_{n})=0$ but $s'_{11}(z_{n})\neq0$, $n=1, 2, \cdots, N$. Recalling the symmetry properties \eqref{Q27} and \eqref{Q28}, we have $s_{22}(z_{n}^{*})=s_{22}(-q_{0}^{2}/z_{n})=s_{11}(-q_{0}^{2}/z_{n}^{*})=0$ if $s_{11}(z)=0$, which gives rise to the set of discrete spectrum
\begin{align}\label{Q29}
\mathbb{Z}=\left\{z_{n}, -q_{0}^{2}/{z_{n}^{*}},
z_{n}^{*}, -q_{0}^{2}/{z_{n}}\right\}_{n=1}^{N},\quad s_{11}(z_{n})=0.
\end{align}

We next study the residue conditions required in the inverse scattering process. If $z_{n}$ is a simple zero of $s_{11}(z)$, the relation will be inferred by the first expression of \eqref{Wr-1}
\begin{align}\label{Q30}
\psi_{+,1}(z_{n})=b_{-}(z_{n})\psi_{-,2}(z_{n}),
\end{align}
here $b_{-}(z_{n})$ is a constant and independent of the variable $x$ and $t$.
For the simple zero $z_{n}\in\mathbb{Z}\cap D^{-}$, the following residue condition can be derived
\begin{align}\label{Q31}
\mathop{Res}_{z=z_{n}}\left[\frac{\psi_{+,1}(z)}{s_{11}(z)}\right]=
\frac{\psi_{-,1}(z_{n})}{s'_{11}(z_{n})}=\frac{b_{+}(z_{n})}
{s'_{11}(z_{n})}\psi_{-,2}(z_{n}).
\end{align}
Furthermore, if $s_{22}(z_{n}^{*})=0$ is a simple zero $(s_{22}(z)\in\mathbb{Z}\cap D^{+})$, the expression of $s_{22}(z)$ implies that
\begin{align}\label{Q32}
\psi_{+,2}(z_{n}^{*})=b_{+}(z_{n}^{*})\psi_{-,1}(z_{n}^{*}),
\end{align}
the another residue condition also can be obtained
\begin{align}\label{Q33}
\mathop{Res}_{z=z_{n}^{*}}\left[\frac{\psi_{+,2}(z)}{s_{22}(z)}\right]=
\frac{\psi_{+,2}(z_{n}^{*})}{s'_{22}(z_{n}^{*})}=
\frac{b_{+}(z_{n}^{*})}{s'_{22}(z_{n}^{*})}\psi_{-,1}(z_{n}^{*}).
\end{align}
For convenience, we introduce the following notation
\begin{align}\label{Q34}
\begin{cases}
C_{-}[z_{n}]=\frac{b_{-}(z_{n})}{s'_{11}(z_{n})},\quad
z_{n}\in\mathbb{Z}\cap D^{-},\\
C_{+}[z_{n}^{*}]=\frac{b_{+}(z_{n}^{*})}{s'_{22}(z_{n}^{*})},\quad
z_{n}^{*}\in\mathbb{Z}\cap D^{+}.
\end{cases}
\end{align}
\begin{cor}According to the symmetries \eqref{Q27} and  \eqref{Q28}, the relationships are derived
\begin{align}\label{Q35}
C_{-}[z_{n}]=-C_{+}^{*}[z_{n}^{*}],\quad C_{-}[z_{n}]=\frac{z_{n}^{2}}{q_{-}^{2}}C_{+}
\left[-\frac{q_{0}^{2}}{z_{n}}\right],\quad z_{n}\in\mathbb{Z}\cap D^{-}.
\end{align}
Further we have
\begin{align}\label{Q36}
C_{-}[z_{n}]=-C_{+}^{*}[z_{n}^{*}]=\frac{z_{n}^{2}}{q_{-}^{2}}
C_{+}\left[-\frac{q_{0}^{2}}{z_{n}}\right]=-\frac{z_{n}^{2}}{q_{-}^{2}}
C_{-}^{*}\left[-\frac{q_{0}^{2}}{z_{n}^{*}}\right],\quad z_{n}\in\mathbb{Z}\cap D^{-}.
\end{align}
\end{cor}
\begin{proof}
The equations \eqref{Q24} and  \eqref{Q30} imply
\begin{align}
\sigma_{0}\psi_{+,2}^{*}(z_{n}^{*})=-b_{-}(z_{n})\sigma_{0}\psi_{-,1}^{*}(z_{n}^{*}),
\end{align}
and combining with \eqref{Q32}, one has $b_{-}(z_{n})=-b_{+}^{*}(z_{n}^{*})$, as well as
\begin{align}
C_{-}[z_{n}]=\frac{b_{-}(z_{n})}{s'_{11}(z_{n})}=\frac{-b_{+}^{*}(z_{n}^{*})}
{(s'_{22}(z_{n}^{*}))^{*}}=-C_{+}[z_{n}^{*}].
\end{align}
Similarly from the expression \eqref{Q25}, we arrive at
\begin{align}
\psi_{+,2}\left(-\frac{q_{0}^{2}}{z_{n}}\right)&=b_{-}(z_{n})
\frac{q_{-}}{q_{+}^{*}}\psi_{-,1}\left(-\frac{q_{0}^{2}}{z_{n}}\right),\\
b_{-}(z_{n})&=\frac{q_{+}^{*}}{q_{-}}b_{+}\left(-\frac{q_{0}^{2}}{z_{n}}\right).
\end{align}
Furthermore
\begin{align}
C_{-}(z_{n})=\frac{b_{z_{n}}}{s'_{11}(z_{n})}=\frac{(q_{+}^{*}/q_{-})
b_{+}\left(-\frac{q_{0}^{2}}{z_{n}}\right)}{(q_{+}^{*}/q_{-}^{*})
(q_{0}^{2}/z_{n}^{2})s'_{22}\left(-\frac{q_{0}^{2}}{z_{n}}\right)}=
\frac{z_{n}^{2}}{q_{-}^{2}}C_{+}\left[-\frac{q_{0}^{2}}{z_{n}}\right],
\end{align}
and
\begin{align}
C_{-}^{*}\left[-\frac{q_{0}^{2}}{z_{n}^{*}}\right]=\frac{b_{-}\left(
-\frac{q_{0}^{2}}{z_{n}^{*}}\right)}{\left(s_{11}^{*}\left(-\frac{q_{0}^{2}}
{z_{n}^{*}}\right)\right)^{'}}=\frac{b_{-}\left(
-\frac{q_{0}^{2}}{z_{n}^{*}}\right)}{s_{22}^{'}\left(-\frac{q_{0}^{2}}{z_{n}}\right)}=
\frac{-b_{+}\left(-\frac{q_{0}^{2}}{z_{n}}\right)}{\frac{z_{n}^{2}}{q_{0}^{2}}
\frac{q_{-}^{*}}{q_{+}^{*}}s'_{11}(z_{n})}=-\frac{q_{-}^{2}}{z_{n}^{2}}C_{-}[z_{n}].
\end{align}
To sum up, we can get this Corollary.
\end{proof}

\section{Inverse scattering problem with the simple poles}
\subsection{A matrix Riemann-Hilbert problem}
\begin{defn}
In \cite{Shabat-1976}, Shabat first proposed the Riemann Hilbert method, seeding an $n\times n$ matrix $M(z)$, given\\
$\blacktriangleright$ An oriented contour $\Sigma$ where $M(z)$ is
 analytic as $z\notin\Sigma$,\\
$\blacktriangleright$ The jump condition $M_{+}(z)=M_{-}(z)V(z)$ for $z\in\Sigma$, here $M_{+}(z)$ and $M_{-}(z)$ represent the non-tangential limits from the left and right of $\Sigma$, respectively. Also $V(z)$ is called the jump matrix,\\
$\blacktriangleright$ The normalization $M(z)\rightarrow I$ as $z\rightarrow\infty$. It should be noted that normalization is necessary for uniqueness, because if $M$ satisfies the conditions of analyticity and normalization, then for any $n\times n$ invertible matrix $C$, $CM$ also satisfies them.
\end{defn}

\centerline{\begin{tikzpicture}[scale=0.5]
\draw[fill] (-4,2) circle [radius=0.035][thick]node[left]{\footnotesize$q(x,0)$};
\draw[fill] (4,2) circle [radius=0.035][thick]node[right]
{\footnotesize$\rho(z,0)$};
\draw[fill] (-4,-2) circle [radius=0.035][thick]node[left]{\footnotesize$q(x,t)$};
\draw[fill] (4,-2) circle [radius=0.035][thick]node[right]
{\footnotesize$\rho(z,t)$};
\draw[->][thick](-3.5,2)--(3.5,2);
\draw[fill] (-0.3,2) node[above]{\footnotesize{Direct scattering}};
\draw[->][thick](4,1.5)--(4,-1.5);
\draw[->][thick](3.5,-2)--(-3.5,-2);
\draw[<-][thick](-4,-1.5)--(-4,1.5);
\draw[fill] (4,0) node[right]{\footnotesize{Time evolution}};
\draw[fill] (0.3,-2) node[above]{\footnotesize{Inverse scattering}};
\end{tikzpicture}}
\centerline{\noindent {\small \textbf{Figure 2.} The inverse scattering method.}}
In what follows, according to the analyticity of Jost function $u_{\pm}(z)$ and scattering matrix $S(z)$ in different regions, we will construct a matrix RH problem, then the form of the solution $q$ can be expressed by solving this RH problem.  So next we need to find out the analytic functions in different regions according to the definition 3.1.
The expression \eqref{M19} implies that
\begin{align}\label{Q37}
\left\{\begin{aligned}
\frac{\psi_{+,1}(x,t;z)}{s_{11}(z)}=\psi_{-,1}(x,t;z)+
\frac{s_{21}(z)}{s_{11}(z)}\psi_{-,2}(x,t;z),\\
\frac{\psi_{+,2}(x,t;z)}{s_{22}(z)}=\frac{s_{12}(z)}{s_{22}(z)}
\psi_{-,1}(x,t;z)+\psi_{-,2}(x,t;z),
\end{aligned}\right.
\end{align}
which is equivalent to
\begin{align}\label{Q38}
\left\{\begin{aligned}
\frac{u_{+,1}(z)}{s_{11}(z)}&=u_{-,1}(z)+\frac{s_{21}(z)}{s_{11}(z)}
e^{2i\theta(z)}u_{-,2}(z),\\
\frac{u_{+,2}(z)}{s_{22}(z)}&=\frac{s_{12}(z)}{s_{22}(z)}
e^{-2i\theta(z)}u_{-,1}(z)+u_{-,2}(z),
\end{aligned}\right.
\end{align}
from the expression \eqref{M16}. Then we define the sectionally  meromorphic matrices
\begin{align}\label{Matr}
M(x,t;z)=\left\{\begin{aligned}
&M^{+}(x,t;z)=\left(u_{-,1}(x,t;z),\frac{u_{+,2}(x,t;z)}{s_{22}(z)}\right), \quad z\in D^{+},\\
&M^{-}(x,t;z)=\left(\frac{u_{+,1}(x,t;z)}{s_{11}(z)},u_{-,2}(x,t;z)\right), \quad z\in D^{-}.
\end{aligned}\right.
\end{align}
satisfying the following Theorem.
\begin{thm}
The matrix RH problem:\\
$\blacktriangleright$ Analyticity: $M(x,t;z)$ is meromorphic in $C\setminus\Sigma$.\\
$\blacktriangleright$ Jump condition:
\begin{align}\label{Jump}
M^{+}(x,t;z)=M^{-}(x,t;z)(\mathbb{I}-G(x,t;z)),\quad z\in\Sigma,
\end{align}
with the jump matrix is
\begin{align*}
G(x,t;z)=e^{i\theta(z)\hat{\sigma}_{3}}\left(
\begin{array}{ccc}
  0 & -\tilde{\rho}(z) \\
  \rho(z) & \rho(z)\tilde{\rho}(z)
\end{array} \right). \notag
\end{align*}\\
$\blacktriangleright$
Asymptotic condition: \begin{align}\label{jianjin}
M^{\pm}(x,t;z)=\left\{
\begin{aligned}
\mathbb{I}+O(1/z), \quad z\rightarrow\infty,\\
-(i/z)\sigma_{3}Q_{-}+O(1), z\rightarrow0,
\end{aligned}
\right.
\end{align}
with the asymptotic condition of the Jost functions and scattering data (For details, please refer to \cite{Biondini-2014})
\begin{align}\label{Q39}
u_{\pm}(x,t;z)=
\left\{
\begin{aligned}
&\mathbb{I}+O(1/z),\quad\quad &z\rightarrow\infty,\\
&-\frac{i}{z}\sigma_{3}Q_{\pm}+O(1),\quad &z\rightarrow0.
\end{aligned}
\right.
\\S(z)=\left\{\begin{aligned}
\mathbb{I}+O(1/z),  \qquad \quad\qquad\qquad z\rightarrow\infty,\\\label{Sjianjin}
diag(q_{-}/q_{+},q_{+}/q_{-})+O(z),\quad z\rightarrow0.
\end{aligned}
\right.
\end{align}
For the convenience, we can be written the discrete spectrum points $\mathbb{Z}$ as
\begin{align}\label{Q40}
\xi_{n}=\left\{\begin{aligned}
&z_{n}, \quad\quad n=1, 2, \cdots,N,\\
-&\frac{q_{0}^{2}}{z_{n-N}^{*}}, \quad n=N+1, N+2, \cdots, 2N,
\end{aligned}\right.
\end{align}
and $\hat{\xi}_{n}=-q_{0}^{2}/\xi_{n}$. Thus the set $\mathbb{Z}=\{\xi_{n}, \hat{\xi}_{n}\}_{n=1}^{2N}$.
\end{thm}
\begin{prop}
The solutions of the matrix RHP \eqref{Jump} can be written as
\begin{align}\label{Q41}
\begin{split}
M(x,t;z)=&\mathbb{I}-\frac{i}{z}\sigma_{3}Q_{+}+\sum_{n=1}^{2N}\frac
{\mathop{Res}\limits_{z=\hat{\xi}_{n}}M^{+}(z)}{z-\xi^{*}_{n}}+\sum_{n=1}^{2N}\frac
{\mathop{Res}\limits_{z=\xi_{n}}M^{-}(z)}{z-\xi_{n}}\\
&+\frac{1}{2\pi i}\int_{\Sigma}\frac{M(x,t;s)^{-}G(x,t;s)}{s-z}\,ds,\quad
z\in\mathbb{C}\setminus\Sigma,
\end{split}
\end{align}
where the $\int_{\Sigma}$ implies the contour shown in Fig. 1.
\end{prop}
\begin{proof}
Subtracting out the asymptotic behavior and the pole contributions to regularize the RH problem \eqref{Jump}, we can solve it by Plemelj's formulae and Cauchy projectors, that is
\begin{align}\label{Q42}
\begin{split}
&M^{+}(x,t;z)-\mathbb{I}+\frac{i}{z}\sigma_{3}Q_{-}-\sum_{n=1}^{2N}\frac
{\mathop{Res}\limits_{z=\hat{\xi}_{n}}M^{+}(z)}{z-\hat{\xi}_{n}}-\sum_{n=1}^{2N}\frac
{\mathop{Res}\limits_{z=\xi_{n}}M^{-}(z)}{z-\xi_{n}}\\=&
M^{-}(x,t;z)-\mathbb{I}+\frac{i}{z}\sigma_{3}Q_{-}-\sum_{n=1}^{2N}\frac
{\mathop{Res}\limits_{z=\hat{\xi}_{n}}M^{+}(z)}{z-\hat{\xi}_{n}}-\sum_{n=1}^{2N}\frac
{\mathop{Res}\limits_{z=\xi_{n}}M^{-}(z)}{z-\xi_{n}}-M^{-}(z)G(z).
\end{split}
\end{align}
The left hand of the equation is analytic in the region $D^{+}$ and tends to $O(1/z)$ as $z\rightarrow\infty$, while the sum of the first four terms for the right hand is analytic in the region $D^{-}$ as well as tends to $O(1/z)$ as $z\rightarrow\infty$. Now recall the Cauchy projectors $P_{\pm}$ over $\Sigma$
\begin{align}\label{Cauchy}
P_{\pm}[f](z)=\frac{1}{2\pi i}\int_{\Sigma}\frac{f(\zeta)}{\zeta-(z\pm i0)}\,d\zeta,
\end{align}
where the $\int_{\Sigma}$ implies the integral along the oriented contour shown in Fig. 1 and the $z\pm i0$ mean that the limit is taken from the left/right of $z$ $ (z\in\Sigma)$ respectively, we get the solution to the equation \eqref{Q42}
\begin{align}\label{Q43}
\begin{split}
M(x,t;z)=&\mathbb{I}-\frac{i}{z}\sigma_{3}Q_{-}+\sum_{n=1}^{2N}\frac
{\mathop{Res}\limits_{z=\hat{\xi}_{n}}M^{+}(z)}{z-\hat{\xi}_{n}}+\sum_{n=1}^{2N}\frac
{\mathop{Res}\limits_{z=\xi_{n}}M^{-}(z)}{z-\xi_{n}}\\
&+\frac{1}{2\pi i}\int_{\Sigma}\frac{M(x,t;\zeta)^{-}G(x,t;\zeta)}{\zeta-z}\,d\zeta,\quad
z\in\mathbb{C}\setminus\Sigma,
\end{split}
\end{align}
via the Plemelj's formulae. As usual the $\int_{\Sigma}$ implies the contour shown in Fig. 1.
\end{proof}
\subsection{Trace formulate and theta condition}
To consider the boundary conditions of $q_{+}$ and $q_{-}$, the trace formulate and theta condition \cite{Faddeev-1987} are introduced via the scattering coefficients $s_{11}(z)$, $s_{22}(z)$ along with the reflection coefficients $\rho(z)$, $\tilde{\rho}(z)$ and discrete spectral points $\mathbb{Z}$. The analyticity of the scattering coefficients $s_{11}(z)$ and $s_{22}(z)$ have been given before, and it is assumed that $z_{n}$ and $-q_{0}^{2}/z_{n}^{*}$ are the zero points of the $s_{11}(z)$, then
\begin{align}\label{trace1}
\beta^{-}(z)=s_{11}(z)\prod_{n=1}^{2N}\frac{(z-z_{n}^{*})(z+q_{0}^{2}/z_{n})}
{(z-z_{n})(z+q_{0}^{2}/z_{n}^{*})},
\end{align}
is analytic in the region $D^{-}$ and has no zeros, as well as the it's asymptotic behavior is same as $s_{11}(z)$ for $z\rightarrow\infty$. Similarly, for the zeros $z_{n}^{*}$ and $-q_{0}^{2}/z_{n}$ of $s_{22}(z)$, we have
\begin{align}\label{trace2}
\beta^{+}(z)=s_{22}(z)\prod_{n=1}^{2N}\frac{(z-z_{n})(z+q_{0}^{2}/z_{n}^{*})}
{(z-z_{n}^{*})(z+q_{0}^{2}/z_{n})},
\end{align}
which is analytic in the region $D^{+}$ and has no zeros, as well as the it's asymptotic behavior is same as $s_{22}(z)$ for $z\rightarrow\infty$. In additional, the asymptotic behavior: $\beta^{\pm}\rightarrow 1$  as $z\rightarrow\infty$ from
\eqref{Sjianjin}, the determinants of $S(z)$ also can be obtained, that is $\det S(z)=s_{11}(z)s_{22}(z)-s_{12}(z)s_{21}(z)=1$. Combining with \eqref{trace1} and \eqref{trace2} yields
\begin{align}\label{det}
\beta^{+}(z)\beta^{-}(z)=\frac{1}{1+\rho(z)\rho^{*}(z^{*})},\quad z\in\Sigma,
\end{align}
which is a scalar, multiplicative RH problem.
Taking logarithms of both sides for \eqref{det}, one has
\begin{align}\label{log}
\log\beta^{+}(z)+\log\beta^{-}(z)=
-\log[1+\rho(z)\rho^{*})(z^{*})], \quad z\in\Sigma.
\end{align}
Moreover, by using the Cauchy projectors and Plemelj's formulae, the solution of \eqref{log} can be written as
\begin{align}
\begin{split}
\log \beta^{-}(z)&=\frac{1}{2\pi i}
\int_{\Sigma}\frac{\log[1+\rho(z)\rho^{*}(z^{*})
(\zeta)]}{\zeta-z}\,d\zeta,\\
\log \beta^{+}(z)&=-\frac{1}{2\pi i}
\int_{\Sigma}\frac{\log[1+\rho(z)\rho^{*}(z^{*})
(\zeta)]}{\zeta-z}\,d\zeta.
\end{split}
\end{align}
Therefore the trace formulae can be constructed by
\begin{align}
s_{11}(z)&=exp\left(\frac{1}{2\pi i}\int_{\Sigma}\frac{\log[1+\rho(\zeta)\rho^{*}
(\zeta^{*})]}{\zeta-z}\,d\zeta\right)\prod_{n=1}^{2N}\frac{(z-z_{n})
(z+q_{0}^{2}/z_{n}^{*})}{(z-z_{n}^{*})(z+q_{0}^{2}/z_{n})},\label{s11}\\
s_{22}(z)&=exp\left(-\frac{1}{2\pi i}\int_{\Sigma}\frac{\log[1+\rho(\zeta)\rho^{*}
(\zeta^{*})]}{\zeta-z}\,d\zeta\right)\prod_{n=1}^{2N}\frac{(z-z_{n}^{*})
(z+q_{0}^{2}/z_{n})}{(z-z_{n})(z+q_{0}^{2}/z_{n}^{*})}.\label{s22}
\end{align}

We next consider the theta condition, the asymptotic phase of the boundary $q_{+}$ and $q_{-}$, based on the obtained results. As $z\rightarrow0$, $s_{11}(z)\rightarrow q_{-}/q_{+}$ from \eqref{Sjianjin}, and it is easy to check that
\begin{align}
\prod_{n=1}^{2N}\frac{(z-z_{n}^{*})(z+q_{0}^{2}/z_{n})}
{(z-z_{n})(z+q_{0}^{2}/z_{n}^{*})}\rightarrow1,\quad z\rightarrow 0.
\end{align}
Then \eqref{s11} implies
\begin{align}
q_{-}/q_{+}=exp\left(\frac{1}{2\pi i}\int_{\Sigma}
\frac{\log[1+\rho(\zeta)\rho^{*}(\zeta^{*})]}{\zeta}\,d\zeta\right),\quad for \quad z\rightarrow 0,
\end{align}
i.e., the theta condition is
\begin{align}\label{theta}
\arg\frac{q_{-}}{q_{+}}=\frac{1}{2\pi}\int_{\Sigma}
\frac{\log[1+\rho(\zeta)\rho^{*}(\zeta^{*})]}{\zeta}\,d\zeta+4\sum_{n=1}^{N}\arg z_{n}.
\end{align}
Obviously, for the reflection-less case, \eqref{theta} can also be expressed as
\begin{align}\label{F5}
\arg\frac{q_{-}}{q_{+}}=\arg q_{-}-\arg q_{+}=4\sum_{n=1}^{N}\arg z_{n}.
\end{align}

\subsection{ A closed  algebraic integral system}
From \eqref{Q42} we know that the residue conditions of $M^{\pm}(z)$ in discrete spectrum points  need to be analyzed to get a closed algebraic integral system for the solution of the RH problem. Therefore
\begin{align}\label{Q44}
\left\{\begin{aligned}
\mathop{Res}_{z=\hat{\xi}_{n}}M^{+}&=(0,C_{+}[\hat{\xi}_{n}]
e^{-2i\theta(\hat{\xi}_{n})}u_{-,1}(\hat{\xi}_{n})),\quad n=1,2,\cdots,2N,\notag \\
\mathop{Res}_{z=\xi_{n}}M^{-}&=(C_{-}[\xi_{n}]e^{2i\theta(\xi_{n})}u_{-,2}(\xi_{n}),0), \quad n=1,2,\cdots,2N.
\end{aligned}\right.
\end{align}
Note that \eqref{Matr} implies only the second (first) column of $M^{+}$ ($M^{-}$) have a simple pole at $z=\xi_{n}$ ($z=\hat{\xi}_{n}$). Thus we get the residue conditions
\begin{align}
\frac{\mathop{Res}\limits_{z=\hat{\xi}_{n}}M^{+}(z)}{z-\hat{\xi}_{n}}+
\frac{\mathop{Res}\limits_{z=\xi_{n}}M^{-}(z)}{z-\xi_{n}}=\left(
\frac{C_{-}[\xi_{n}]e^{2i\theta(\xi_{n})}}{z-\xi_{n}}
u_{-,2}(\xi_{n}),
\frac{C_{+}[\hat{\xi}_{n}]e^{-2i\theta(\hat{\xi}_{n})}}{z-\hat{\xi}_{n}}
u_{-,1}(\hat{\xi}_{n})\right),
\end{align}
furthermore
\begin{align}\label{Q45}
\begin{split}
u_{-,2}(x,t;\xi_{s})=\left(\begin{array}{cc}
                       -iq_{-}/\xi_{s} \\
                        1
                     \end{array}\right)&+\sum_{n=1}^{2N}
\frac{C_{+}[\hat{\xi}_{n}]e^{-2i\theta(x,t;\hat{\xi}_{n})}}{\xi_{n}-\hat{\xi}_{n}}
u_{-,1}(x,t;\hat{\xi}_{n})\\
&+\frac{1}{2\pi i}\int_{\Sigma}\frac{(M^{-}G)_{2}(x,t;\zeta)}{\zeta-\xi_{s}}\,d\zeta,
\quad s=1, 2, \cdots, 2N.
\end{split}
\end{align}
The expression \eqref{Q25} implies
\begin{align}\label{Q46}
u_{-,2}(\xi_{s})=-\frac{iq_{-}}{\hat{\xi}_{s}}u_{-,1}(\hat{\xi}_{n}),\quad s=1, 2, \cdots, 2N,
\end{align}
then substituting it into \eqref{Q45}, one has
\begin{align}\label{Q47}
\sum_{n=1}^{2N}\left(\frac{C_{+}[\hat{\xi}_{n}]e^{-2i\theta(\hat{\xi}_{n})}}
{\xi_{s}-\hat{\xi}_{n}}+\frac{iq_{-}}{\xi_{s}}\delta_{sn}\right)
u_{-,1}(\hat{\xi}_{n})+\left(\begin{array}{cc}
                      -\frac{iq_{-}}{\xi_{s}}  \\
                        1
                     \end{array}\right)+
  \frac{1}{2\pi i}\int_{\Sigma}\frac{(M^{-}G)_{2}(\xi)}{\xi-\xi_{s}}\,d\xi=0,
\end{align}
with the function $\delta_{sn}$, the Kronecker delta function, means that $\delta_{sn}=\left\{\begin{aligned}
1, \quad s=n, \\
0, \quad s\neq n.
\end{aligned}
\right.
$.
Our ultimate goal is to express the solution $M(z)$ of RH problem according to the scattering data, so we need to know the expression of residue condition \eqref{Q44} , and then solve the expression of $u_{-,1}(\hat{\xi}_{n})$ and $u_{-,2}(\xi_{n})$ equally. From system \eqref{Q47}, we can get $2N$ unknowns of $u_{-,1}(\hat{\xi}_{n})$, and then the $2N$ expressions of $u_{-,2}(\xi_{n})$ can be derived by taking them into the equation \eqref{Q46}. Finally, according to equation \eqref{Q43}, one has the formal expression of $M(z)$ for the RH problem.

Next, we will express the potential with single pole  of the equation KE  under the condition of non-zero boundary value according to the above results.
\begin{prop}
The formal of potential $q$ for the KE can be written as
\begin{align}\label{Q48}
q(x,t)=-q_{-}-i\sum_{n=1}^{2N}C_{+}[\hat{\xi}_{n}]e^{-2i\theta(\hat{\xi}_{n})}
u_{-,11}(x,t;\hat{\xi}_{n})+\frac{1}{2\pi}\int_{\Sigma}(M^{-}G)_{12}(x,t;\xi)\,d\xi,
\end{align}
with $\xi_{n}$ is defined by \eqref{Q40}, and $u_{-,11}(x,t;\hat{\xi}_{n})$ depends on
\begin{align}\label{u11}
\sum_{n=1}^{2N}\left(\frac{C_{+}[\hat{\xi}_{n}]e^{-2i\theta(\hat{\xi}_{n})}}
{\xi_{s}-\hat{\xi}_{n}}+\frac{iq_{-}}{\xi_{s}}\delta_{sn}\right)
u_{-,11}(x,t;\hat{\xi}_{n})- \frac{iq_{-}}{\xi_{s}} +
  \frac{1}{2\pi i}\int_{\Sigma}\frac{(M^{-}G)_{1,2}(x,t;\xi)}{\xi-\xi_{s}}\,d\xi=0.
  \end{align}
\end{prop}
\begin{proof}
The asymptotic behavior of $M(z)$ satisfies
\begin{align}\label{Taylor}
M(x,t;z)=\mathbb{I}+\frac{1}{z}M^{(1)}(x,t;z)+O(1/z^{2}), \quad z\rightarrow\infty,
\end{align}
and
\begin{align}\label{Q49}
\begin{split}
M^{(1)}(z)=-i\sigma_{3}Q_{-}+\sum_{n=1}^{2N}&\left(
C_{-}[\xi_{n}]e^{2i\theta(\xi_{n})}u_{-,2}(\xi_{n}),
C_{+}[\hat{\xi}_{n}]e^{-2i\theta(\hat{\xi}_{n})}u_{-,1}(\hat{\xi}_{n})\right)\\
&-\frac{1}{2\pi i}\int_{\Sigma}(M^{-}G)_{1,2}(x,t;\xi)\,d\xi=0.
\end{split}
\end{align}
In additional, $M(z)e^{i\theta(z)\sigma_{3}}$ such that \eqref{M7}, thus we obtain
\begin{align}\label{Q50}
M_{x}(x,t;z)+M(x,t;z)\left(\frac{1}{2}i\sigma_{3}z+\frac{1}{2z}iq_{0}^{2}
\sigma_{3}\right)=\left(\frac{1}{2}i\sigma_{3}z-\frac{1}{2z}iq_{0}^{2}
\sigma_{3}+\tilde{Q}\right)M(x,t;z).
\end{align}
Plugging \eqref{Taylor} into \eqref{Q50} and collecting the coefficient of $z^{0}$, then we can get
\begin{align}\label{Q51}
q(x,t)=-i(M^{(1)})_{12}.
\end{align}
The Proposition can be derived by the equation \eqref{Q49}.
\end{proof}

For our convenience, let's make the following notions:
$T=(t_{sj})_{(2N)\times(2N)}$, $\omega=(\omega_{j})_{(2N)\times1}$, and $\nu=(\nu_{j})_{(2N)\times1}$ with $t_{sj}=\frac{\omega_{j}}{\xi_{s}-\hat{\xi}_{j}}+\nu_{s}\delta_{sj}$, $\omega_{j}=C_{+}[\hat{\xi}_{j}]e^{-2i\theta(\hat{\xi_{j}})}$, and $\nu_{j}=\frac{iq_{-}}{\xi_{j}}$, besides $x_{n}=u_{-,11}(x,t;\hat{\xi}_{n})$.
Now let's consider the case of non-reflection, that is  $G=0$. From \eqref{u11}, we know that $x=(u_{-,11}(x,t;\hat{\xi}_{n}))_{(2N)\times1}=T^{-1}\nu$, as well as the following Proposition:
\begin{prop}
The non-zero boundary value form solution of the equation \eqref{M7}  without scattering coefficients can be written as
\begin{align}\label{Jie-Q}
q(x,t)=-q_{-}+i\frac{\det\left(\begin{array}{cc}
                      T & \nu \\
                      \omega^{T} & 0
                    \end{array}\right)}{\det T}.
\end{align}
\end{prop}
\subsection{Soliton solutions}
In this section, the propagation behavior characteristics of the solution for the  focusing Kundu-Eckhaus equation will be presented in combination with the specific parameters. Generally, the expressions of these solutions are very complex. In order to be consistent, we will omit the expressions of these solutions, only give the propagation image and make a brief analysis.

\noindent \textbf {Case 1.}
With considering the discrete eigenvalue is purely imaginary, that is $z=iZ$ $(Z>1)$,
there is no phase different, further the solutions can be given with appropriate paeameters with taking $N=1$

{\rotatebox{0}{\includegraphics[width=2.75cm,height=2.5cm,angle=0]{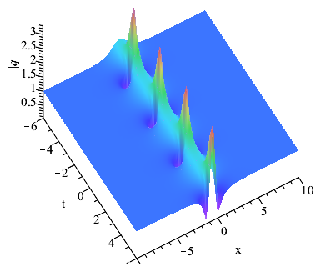}}}
{\rotatebox{0}{\includegraphics[width=2.75cm,height=2.5cm,angle=0]{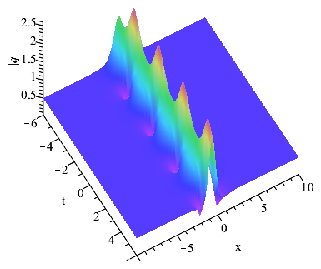}}}
{\rotatebox{0}{\includegraphics[width=2.75cm,height=2.5cm,angle=0]{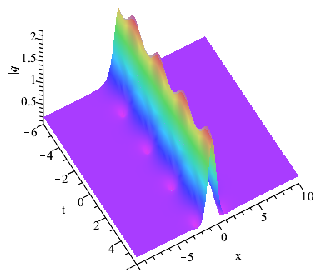}}}
{\rotatebox{0}{\includegraphics[width=2.75cm,height=2.5cm,angle=0]{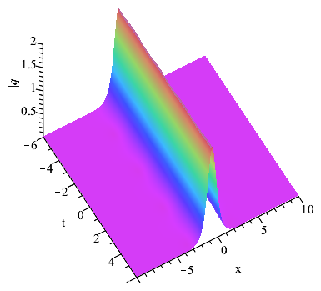}}}

 $\qquad\quad(\textbf{a1})\qquad \ \quad\quad\qquad(\textbf{b1})
\ \qquad\quad\quad\quad\qquad(\textbf{c1})\qquad\quad\quad\qquad(\textbf{d1})
\ \qquad$\\

{\rotatebox{0}{\includegraphics[width=2.75cm,height=2.15cm,angle=0]{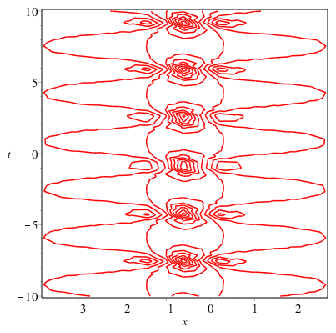}}}
{\rotatebox{0}{\includegraphics[width=2.75cm,height=2.15cm,angle=0]{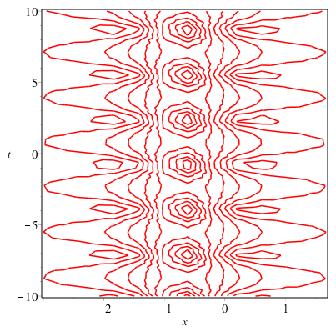}}}
{\rotatebox{0}{\includegraphics[width=2.75cm,height=2.15cm,angle=0]{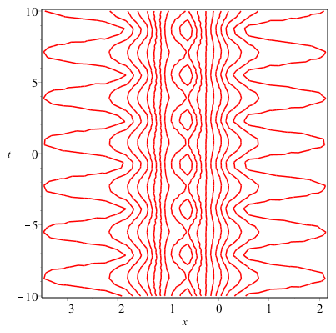}}}
{\rotatebox{0}{\includegraphics[width=2.75cm,height=2.15cm,angle=0]{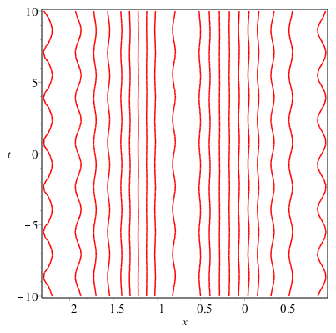}}}

 $\quad\qquad(\textbf{a2})\qquad \ \quad\qquad\qquad(\textbf{b2})
\ \qquad\quad\quad\quad\quad(\textbf{c2})\quad\quad\qquad\qquad(\textbf{d2})
\ \qquad$\\
\noindent { \small \textbf{Figure 3.} The  breather wave solutions  with the fixed parameters $N=1$, $z_{1}=2i$, $C_{-}[z_{1}]=1$. $\textbf{(a1)}$: the breather wave solution with the  $q_{-}=1$; $\textbf{(b1)}$: the breather wave solution with the  $q_{-}=0.5$; $\textbf{(c1)}$: the breather wave solution with the $q_{-}=0.2$; $\textbf{(d1)}$: the bright soliton solution with the  $q_{-}=0.02$; $\textbf{(a2)-(d2)}$: correspond to the contour respectively.}

For $z_{1}=2i$, the asymptotic phase is $2\pi$ from the theta condition \eqref{F5}, which implies that there is no phase difference in this case. From Fig. 3(a1), it can be seen directly that the solution is homoclinic in $x$-axis and periodic in $t$-axis, which is called Kuznetsov-Ma breather solution proposed by Kuznetsov and Ma \cite{Kuznetsov-1977,Ma-1979}. With the decrease of the boundary condition  $q_{-}$, the periodic behavior of the solution gradually moves up until the boundary value tends to zero, then the solution tends to bright soliton solutions shown in Fig. 3(d1).

\noindent \textbf {Case 2.} When the eigenvalue is complex parameters, we give the propagation behavior of the solution \eqref{Jie-Q}, which is compared with Figure 3 and analyzed briefly. It is worth noting that the solution is no longer homoclinic  in $x$-axis and period in $t$-axis.

{\rotatebox{0}{\includegraphics[width=2.85cm,height=2.8cm,angle=0]{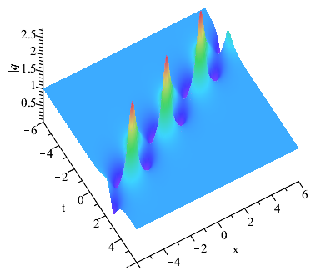}}}
{\rotatebox{0}{\includegraphics[width=2.85cm,height=2.15cm,angle=0]{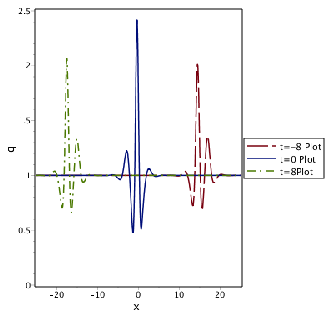}}}
{\rotatebox{0}{\includegraphics[width=2.85cm,height=2.8cm,angle=0]{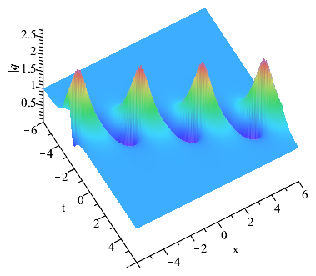}}}
{\rotatebox{0}{\includegraphics[width=2.85cm,height=2.15cm,angle=0]{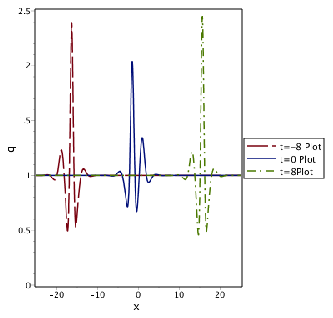}}}

 $\quad\qquad(\textbf{a})\quad \ \qquad\qquad\qquad\quad(\textbf{b})
\ \qquad\quad\quad\quad\qquad\quad(\textbf{c})\qquad\qquad\qquad\qquad(\textbf{d})
\ \qquad$\\
\noindent { \small \textbf{Figure 4.} The breather wave soliton solutions  with the fixed parameters $N=1$,  $C_{-}[z_{1}]=1$ and $q_{-}=1$. $\textbf{(a)}$: the breather wave solution with the $z_{1}=2e^{i\pi/4}$; $\textbf{(b)}$: exhibits the propagation path of solutions at different times; Fig. $\textbf{(c)}$: the breather wave solution with the $z_{1}=2e^{3i\pi/4}$;  $\textbf{(d)}$: denotes the propagation path of solutions at different times.}

The phenomenon, the behavior of the solution is neither parallel to the $x$-axis nor to the $t$-axis, can be seen with comparing the Fig. 4 with Fig. $3(a1)$, which is called non-stationary solutions. The reason is that for the discrete spectral point $z_{1}$ in this case, there exits the phase different. Note that the auxiliary angle a is the critical point, which represents Kuznetsov-Ma breather solution. When the auxiliary angle is less than or greater than the critical point, this will change the propagation behavior of the solution.  In additional, from the propagation path of Fig. $4(b,c)$, we can see that the behavior of solution is symmetrical about time $t=0$.

{\rotatebox{0}{\includegraphics[width=2.85cm,height=2.8cm,angle=0]{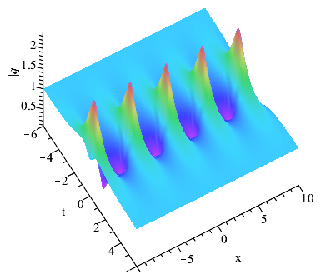}}}
{\rotatebox{0}{\includegraphics[width=2.35cm,height=2.15cm,angle=0]{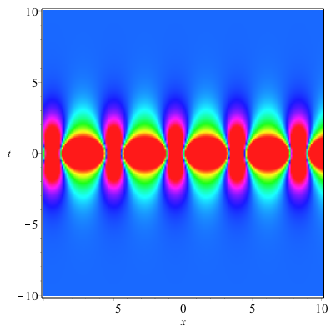}}}
{\rotatebox{0}{\includegraphics[width=2.85cm,height=2.8cm,angle=0]{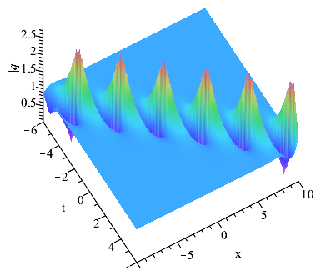}}}
{\rotatebox{0}{\includegraphics[width=2.35cm,height=2.15cm,angle=0]{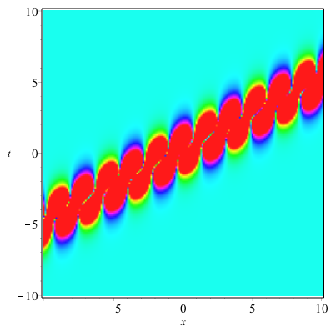}}}

 $\qquad\quad(\textbf{a})\quad \ \qquad\qquad\qquad\quad(\textbf{b})
\ \qquad\quad\quad\quad\qquad(\textbf{c})\qquad\qquad\qquad\qquad(\textbf{d})
\ \qquad$\\
\noindent { \small \textbf{Figure 5.} The breather soliton solutions  with the fixed parameters $N=1$,  $C_{+}[z_{1}]=1$. Fig. $\textbf{(a)}$: the space-breather solution with the $q_{-}=1$, $z_{1}=e^{i\pi/4}$ and $\delta_{2}=0$; $\textbf{(b)}$: presents the density of the space-breather solution; Fig. $\textbf{(c)}$: the breather solution with the $q_{-}=1$, $z_{1}=\frac{1}{2}e^{i\pi/4}$ and $\delta_{2}=1$; $\textbf{(d)}$: denotes the density of the breather solution.}

Obviously, from the observation of Fig. $5(a)$, it is not hard find that the solution changes periodically on the $x$-axis, and is homoclinic on the $t$-axis, which is just the opposite of Fig. $3(a1)$. The graph of its density also shows the behavior of the solution more clearly. In Fig. $5(c)$, the solution under the condition of the parameter at this time is also called the non-stationary solution.

In what follows, taking $N=2$ the different kinds solutions are listed. By selecting different discrete eigenvalues, the influence of them on solution propagation is observed.

{\rotatebox{0}{\includegraphics[width=3.6cm,height=3.0cm,angle=0]{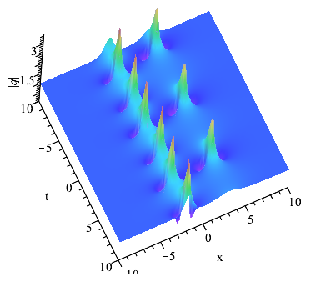}}}
{\rotatebox{0}{\includegraphics[width=3.3cm,height=2.55cm,angle=0]{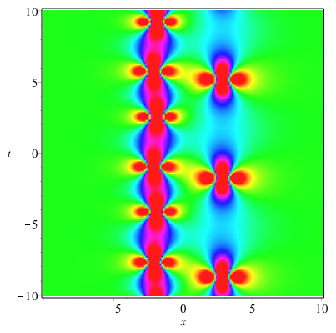}}}
\quad
{\rotatebox{0}{\includegraphics[width=3.6cm,height=2.75cm,angle=0]{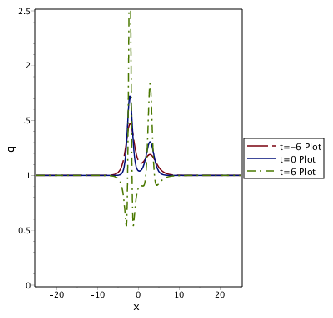}}}

$\qquad\qquad\quad(\textbf{a1})\quad \ \quad\qquad\qquad\qquad\qquad(\textbf{a2})
\qquad\qquad\qquad\quad\qquad(\textbf{a3})$\\

{\rotatebox{0}{\includegraphics[width=3.6cm,height=3.0cm,angle=0]{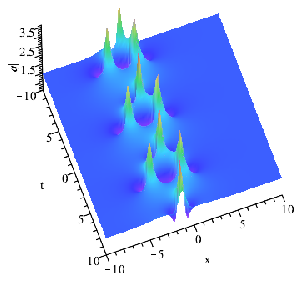}}}
{\rotatebox{0}{\includegraphics[width=3.3cm,height=2.55cm,angle=0]{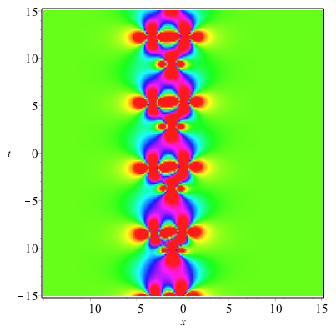}}}
\quad
{\rotatebox{0}{\includegraphics[width=3.6cm,height=2.75cm,angle=0]{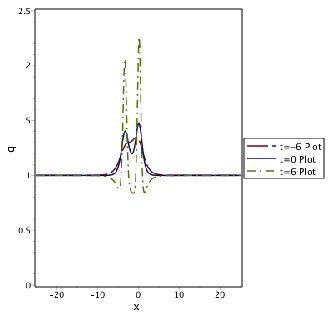}}}

$\qquad\qquad\quad(\textbf{b1})\quad \ \quad\qquad\qquad\qquad\qquad(\textbf{b2})
\qquad\qquad\qquad\quad\qquad(\textbf{b3})$\\

{\rotatebox{0}{\includegraphics[width=3.6cm,height=3.0cm,angle=0]{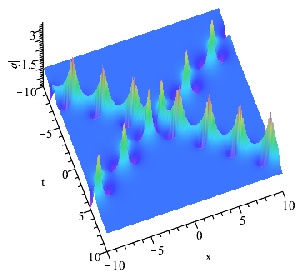}}}
{\rotatebox{0}{\includegraphics[width=3.3cm,height=2.55cm,angle=0]{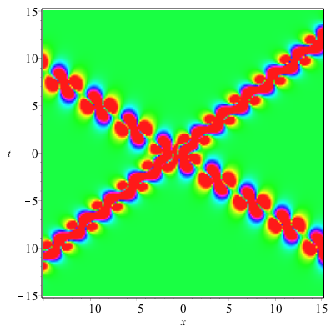}}}
\quad
{\rotatebox{0}{\includegraphics[width=3.6cm,height=2.75cm,angle=0]{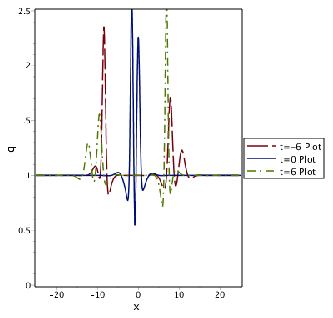}}}

$\qquad\qquad\quad(\textbf{c1})\quad \ \quad\qquad\qquad\qquad\qquad(\textbf{c2})
\qquad\qquad\qquad\quad\qquad(\textbf{c3})$\\

\noindent { \small \textbf{Figure 6.} The breather-type solution with the parameters $C_{-}[z_{1}]=C_{-}[z_{2}]=1$, and $q_{-}=1$.
$\textbf{(a1, b1, c1)}$ show the solution with $z_{1}=\frac{1}{2}i$, $z_{2}=\frac{3}{2}i$;  $z_{1}=2i$, $z_{2}=\frac{3}{2}i$; $z_{1}=-1+2i$, $z_{2}=1+\frac{3}{2}i$.
$\textbf{(a2, b2, c2)}$ present density of different kinds solutions;
$\textbf{(a3,b3,c3)}$ display the propagation behavior of solutions at different times. the breather-breather solution with $t=-8, t=0, t=8$.}
\section{The focusing KE equation with NZBCs: double poles}
In the previous work, we have discussed the case that the scattering coefficients are of single poles in the discrete spectrum points. On this basis, we further study the case that the discrete spectrum points are of double poles, that is, if $z_{n}$ is of double poles for the scattering coefficient $s_{11}(z)$, then $s_{11}(z_{n})=s'_{11}(z_{n})=0$ along with $s''_{11}(z_{n})\neq0$ (Note that and in what follows $'$ denotes the derivative with respect to $z$). These two cases have similar treatment methods, but there are also differences, such as residue conditions, trace formula, etc. It is worth noting that in the following residue calculation, we have to calculate the coefficient of the term of order $-2$.

Now we assume that the scattering coefficient $s_{11}(z)$ and $s_{22}(z)$ are of double poles in the discrete spectrum points $\mathbb{Z}$, which means
\begin{align}\label{T1}
\left\{\begin{aligned}
s_{11}(z_{n})=s'_{11}(z_{n})=0, s''_{11}(z_{n})\neq0,\quad
z_{n}\in\mathbb{Z}\cap D^{-},\\
s_{22}(z_{n})=s'_{22}(z_{n})=0, s''_{22}(z_{n})\neq0,\quad
z_{n}\in\mathbb{Z}\cap D^{+}.
\end{aligned} \right.
\end{align}
In order to calculate the residue, we give the following proposition in similar \cite{Pichler-2017}
\begin{prop}
If the function $f$ and $g$ are analytic in a  complex  region $\Omega\in\mathbb{C}$ , and $g$ is of a double zero $z_{n}\in\Omega$ and $f(z_{n})\neq0$, then the residue condition of $f/g$ can be derived by the Laurent expansion at $z=z_{n}$, namely
\begin{align}\label{T2}
\mathop{Res}_{z=z_{n}}\left[\frac{f}{g}\right]=\frac{2f'(z_{n})}{g''(z_{n})}-
\frac{2f(z_{n})g'''(z_{n})}{3(g''(z_{n}))^{2}},\quad
\mathop{P_{-2}}_{z=z_{n}}\left[\frac{f}{g}\right]=\frac{2f(z_{n})}{g''(z_{n})}.
\end{align}
\end{prop}
Recall the symmetries \eqref{Q27} and \eqref{Q28}, following that if $s_{11}(z_{n})=0$, then
\begin{align}\label{T3}
s_{22}(z_{n}^{*})=s_{22}(-q_{0}^{2}/z_{n})=s_{11}(-q_{0}^{2}/z_{n}^{*})=0.
\end{align}
Similarly, if $s_{11}(z)$ and $s_{22}(z)$ satisfy the condition \eqref{T1}, the equation \eqref{Q30} and \eqref{Q32} still holds (Note that the previous notion $b_{-}(z_{n})$ and $b_{+}(z_{n}^{*})$ still used for convenience, but it may be different from the previous one), then the Proposition can be derived:
\begin{prop}
 Two pairs of linear relations:
\addtocounter{equation}{1}
\begin{align}\label{T4}
\psi_{+,1}(x,t;z_{n})&=b_{-}(z_{n})\psi_{-,2}(x,t;z_{n}),\tag{\theequation a}\\
\psi'_{+,1}(x,t;z_{n})&=d_{-}(z_{n})\psi_{-,2}(x,t;z_{n})+
b_{-}(z_{n})\psi'_{-,2}(x,t;z_{n}),\tag{\theequation b}\\ \label{t1}
\psi_{+,2}(x,t;z_{n}^{*})&=b_{+}(z_{n}^{*})\psi_{-,1}(x,t;z_{n}^{*}),
\tag{\theequation c}\\
\psi'_{+,2}(x,t;z_{n}^{*})&=d_{+}(z_{n}^{*})\psi_{-,1}(x,t;z_{n}^{*})+
b_{+}(z_{n}^{*})\psi'_{-,1}(x,t;z_{n}^{*}).\tag{\theequation d} \label{t3}
\end{align}
\end{prop}
\begin{proof}
By deriving the equation \eqref{Wr-1} about $z$, we can get
\begin{align}\label{T5}
[s_{11}(z)/\gamma]=Wr[\psi'_{+,1}(z),\psi_{-,2}(z)]+
Wr[\psi_{+,1}(z),\psi'_{-,2}(z)],
\end{align}
and $s_{11}(z)$ such that \eqref{T1}, thus
\begin{align}\label{T6}
Wr[\psi'_{+,1}(z_{n})-b_{-}(z_{n})\psi'_{-,2}(z_{n}),\psi_{-,2}(z_{n})]=0,
\end{align}
which means that there exist a constant $d_{-}(z_{n})$ satisfying
\begin{align*}
\psi'_{+,1}(x,t;z_{n})&=d_{-}(z_{n})\psi_{-,2}(x,t;z_{n})+
b_{-}(z_{n})\psi'_{-,2}(x,t;z_{n}).
\end{align*}
Another  case can be proved similarly, so the proof of Proposition is completed.
\end{proof}
Due to $z_{n}$ is a double zero of $s_{11}(z)$, from \eqref{T2} the residue conditions can be derived
\addtocounter{equation}{1}
\begin{align}\label{Res1}
\mathop{P_{-2}}_{z=z_{n}}\left[\frac{\psi_{+,1}(x,t;z)}{s_{11}(z)}\right]&=
\frac{2\psi_{+,1}(z_{n})}{s''_{11}(z_{n})}=\frac{2b_{-}(z_{n})}{s''_{11}(z_{n})}
\psi_{-,2}(z_{z_{n}})=A_{-}[z_{n}]\psi_{-,2}(z_{n}), \tag{\theequation a}\\
\mathop{Res}_{z=z_{n}}\left[\frac{\psi_{+,1}(x,t;z)}{s_{11}(z)}\right]
&=\frac{2b_{-}(z_{n})}{s''_{11}(z_{n})}\left[\psi'_{-,2}(z_{n})+
\left(\frac{d_{-}(z_{n})}{b_{-}(z_{n})}-\frac{s'''_{11}(z_{n})}{3s''_{11}(z_{n})}
\right)\psi_{-,2}(z_{n})\right].\tag{\theequation b} \label{t2}
\end{align}
Similarly, for $z_{n}^{*}$ is a double zero of $s_{22}(z)$, one has
\addtocounter{equation}{1}
\begin{align}\label{Res2}
\mathop{P_{-2}}_{z=z_{n}^{*}}\left[\frac{\psi_{+,2}(x,t;z)}{s_{22}(z)}\right]&=
\frac{2\psi_{+,2}(z_{n}^{*})}{s''_{22}(z_{n}^{*})}=
\frac{2b_{+}(z_{n}^{*})}{s''_{22}(z_{n}^{*})}\psi_{-,1}(z_{z_{n}}^{*})
=A_{+}[z_{n}^{*}]\psi_{-,1}(z_{n}^{*}), \tag{\theequation a}\\
\mathop{Res}_{z=z_{n}^{*}}\left[\frac{\psi_{+,2}(x,t;z)}{s_{22}(z)}\right]
&=\frac{2b_{+}(z_{n}^{*})}{s''_{22}(z_{n}^{*})}\left[\psi'_{-,1}(z_{n}^{*})+
\left(\frac{d_{+}(z_{n}^{*})}{b_{+}(z_{n}^{*})}-
\frac{s'''_{22}(z_{n}^{*})}{3s''_{22}(z_{n}^{*})}
\right)\psi_{-,1}(z_{n}^{*})\right].\tag{\theequation b} \label{t4}
\end{align}
Denoting
\begin{align}\label{T7}
A_{-}[z_{n}]=\frac{2b_{-}(z_{n})}{s''_{11}(z_{n})},\quad
B_{-}[z_{n}]=\frac{d_{-}(z_{n})}{b_{-}(z_{n})}-
\frac{s'''_{11}(z_{n})}{3s''_{11}(z_{n})},\\
A_{+}[z_{n}^{*}]=\frac{2b_{+}(z_{n}^{*})}{s''_{22}(z_{n}^{*})},\quad
B_{+}[z_{n}^{*}]=\frac{d_{+}(z_{n}^{*})}{b_{+}(z_{n}^{*})}-
\frac{s'''_{22}(z_{n}^{*})}{3s''_{22}(z_{n}^{*})}.
\end{align}
the expressions \eqref{t2} and \eqref{t4} can be written as
\begin{align}\label{T8}
\mathop{Res}_{z=z_{n}}\left[\frac{\psi_{+,1}(x,t;z)}{s_{11}(z)}\right]=
A_{-}[z_{n}]\left[\psi'_{-,2}(x,t;z_{n})+B_{-}[z_{n}]\psi_{-,2}(x,t;z_{n})\right],\\
\mathop{Res}_{z=z_{n}^{*}}\left[\frac{\psi_{+,2}(x,t;z)}{s_{22}(z)}\right]=
A_{+}[z_{n}^{*}]\left[\psi'_{-,1}(x,t;z_{n}^{*})+
B_{+}[z_{n}^{*}]\psi_{-,1}(x,t;z_{n}^{*})\right].
\end{align}
Next, we will find the relationship between the data $A_{\pm}$ and $B_{\pm}$ according to the symmetry of scattering data and Jost function, then
\begin{prop}
\begin{align}
\left\{\begin{aligned}
A_{-}[z_{n}]&=-A_{+}^{*}[z_{n}^{*}]=\frac{z_{n}^{4}q_{-}^{*}}
{q_{0}^{4}q_{-}}A_{+}\left(-\frac{q_{0}^{2}}{z_{n}}\right)=-\frac{z_{n}^{4}q_{-}^{*}}
{q_{0}^{4}q_{-}}A_{-}\left(-\frac{q_{0}^{2}}{z_{n}^{*}}\right),
\quad\quad z_{n}\in \mathbb{Z}\cap D^{-},\\
B_{-}[z_{n}]&=B_{+}^{*}[z_{n}^{*}]=\frac{q_{0}^{2}}
{z_{n}^{2}}B_{+}\left(-\frac{q_{0}^{2}}{z_{n}}\right)+\frac{2}{z_{n}}=\frac{q_{0}^{2}}
{z_{n}^{2}}B_{-}^{*}\left(-\frac{q_{0}^{2}}{z_{n}^{*}}\right)+\frac{2}{z_{n}}.
\quad z_{n}\in \mathbb{Z}\cap D^{-}.
\end{aligned} \right.
\end{align}
\end{prop}
\begin{proof}
From the expression \eqref{Q28}, the first, second and third derivatives of $s_{11}(z)$ can be infered with respect to variable $z$, therefore
\begin{align}\label{T9}
A_{-}[z_{n}]&=\frac{2b_{-}(z_{n})}{s'_{11}(z_{n})}=-\frac{2b_{-}^{*}(z_{n}^{*})}
{(s''_{22}(z_{n}^{*}))^{*}}=-A_{+}^{*}[z_{n}^{*}],\\
A_{-}[z_{n}]&=\frac{2b_{-}(z_{n})}{s'_{11}(z_{n})}=\frac{2(q_{+}^{*}/q_{-})
b_{+}\left(-\frac{q_{0}^{2}}{z_{n}}\right)}{(q_{+}^{*}/q_{-}^{*})
(q_{0}^{4}/z_{n}^{4})s''_{22}\left(-\frac{q_{0}^{2}}{z_{n}}\right)}=
\frac{z_{n}^{4}q_{-}^{*}}{q_{-}^{4}q_{-}}A_{+}\left(-\frac{q_{0}^{2}}{z_{n}}\right),\\
B_{-}[z_{n}]&=\frac{d_{-}^{*}(z_{n}^{*})}{b_{-}^{*}(z_{n}^{*})}-
\frac{(s'''_{11}(z_{n}^{*}))^{*}}{(s''_{11}(z_{n}^{*}))^{*}}=B_{+}^{*}[z_{n}^{*}].
\end{align}
Combine \eqref{Q25} with \eqref{t1}, we have
\begin{align}\label{T10}
\psi'_{+,2}(z)=\frac{q_{-}q_{0}^{2}}{z^{2}q_{+}^{*}}d_{-}\left(\frac{-q_{0}^{2}}
{z}\right)\psi_{-,1}(z)+\frac{q_{-}}{q_{+}^{*}}b_{-}\left(\frac{-q_{0}^{2}}
{z}\right)\psi'_{-,1}(z).
\end{align}
The coefficients of comparison \eqref{T10} and \eqref{t3} are
\begin{align}
\left\{\begin{aligned}
\frac{q_{-}q_{0}^{2}}{q_{+}^{*}z_{2}}d_{-}\left(\frac{-q_{0}^{2}}
{z}\right)=d_{+}(z),\\
\frac{q_{-}}{q_{+}^{*}}b_{-}\left(\frac{-q_{0}^{2}}
{z}\right)=b_{+}(z),
\end{aligned}
\right.\Rightarrow
\left\{\begin{aligned}
d_{-}(z)&=\frac{q_{+}^{*}q_{0}^{2}}{q_{-}z^{2}}d_{+}\left(\frac{-q_{0}^{2}}
{z}\right),\\
b_{-}(z)&=\frac{q_{+}^{*}}{q_{-}}b_{+}\left(\frac{-q_{0}^{2}}
{z}\right).\end{aligned} \right.
\end{align}
Further
\begin{align}
\begin{split}
B_{-}[z_{n}]&=\frac{d_{-}(z_{n})}{b_{-}(z_{n})}-\frac{s'''_{11}(z_{n})}{3s''_{11}(z_{n})}
=\frac{\frac{q_{+}^{*}q_{0}^{2}}{q_{-}z_{n}^{2}}d_{+}\left(-\frac{q_{0}^{2}}{z_{n}}
\right)}{\frac{q_{+}^{*}}{q_{-}}b_{+}\left(-\frac{q_{0}^{2}}{z_{n}}\right)}-
\frac{\frac{q_{+}^{*}q_{0}^{6}}{q_{-}^{*}z_{n}^{6}}s'''_{22}
\left(-\frac{q_{0}^{2}}{z_{n}}\right)-\frac{6q_{+}^{*}q_{0}^{4}}{q_{-}^{*}z_{n}^{5}}
s''_{22}\left(-\frac{q_{0}^{2}}{z_{n}}\right)}{3\frac{q_{+}^{*}q_{0}^{4}}
{q_{-}^{*}z_{n}^{4}}s''_{22}\left(-\frac{q_{0}^{2}}{z_{n}}\right)}\\
&=\frac{q_{0}^{2}}{z_{n}^{2}}\frac{d_{+}\left(\frac{-q_{0}^{2}}
{z_{n}}\right)}{b_{+}\left(\frac{-q_{0}^{2}}{z_{n}}\right)}-\frac{q_{0}^{2}}{z_{n}^{2}}
\frac{s'''_{22}\left(\frac{-q_{0}^{2}}{z_{n}}\right)}
{3s''_{22}\left(\frac{-q_{0}^{2}}{z_{n}}\right)}+\frac{2}{z_{n}}=
\frac{q_{0}^{2}}{z_{n}^{2}}B_{-}\left[\frac{-q_{0}^{2}}{z_{n}}\right]+\frac{2}{z_{n}}.
\end{split}
\end{align}
In conclusion, the proposition is proved.
\end{proof}

In addition, the above residue conditions can be converted to via \eqref{M16} and \eqref{Q40}
\addtocounter{equation}{1}
\begin{align}\label{Residue1}
\mathop{P_{-2}}_{z=\xi_{n}}M_{1}^{-}(z)=\mathop{P_{-2}}_{z=\xi_{n}}
\left[\frac{u_{+,1}(z)}{s_{11}(z)}\right]&=\frac{2u_{+,1}(\xi_{n})}
{s''_{11}(\xi_{n})}=\frac{2b_{-}(\xi_{n})}{s''_{11}(\xi_{n})}e^
{2i\theta(\xi_{n})}u_{-,2}(\xi_{n}),\tag{\theequation a}\\
\mathop{P_{-2}}_{z=\hat{\xi}_{n}}M_{2}^{+}(z)=\mathop{P_{-2}}_{z=\hat{\xi}_{n}}
\left[\frac{u_{+,2}(z)}{s_{22}(z)}\right]&=\frac{2u_{+,2}(\hat{\xi}_{n})}
{s''_{22}(\hat{\xi}_{n})}=\frac{2b_{+}(\hat{\xi}_{n})}{s''_{22}(\hat{\xi}_{n})}e^
{-2i\theta(\hat{\xi}_{n})}u_{-,1}(\hat{\xi}_{n}),\tag{\theequation b}\\
\mathop{Res}_{z=\xi_{n}}M_{1}^{-}(z)=\mathop{Res}_{z=\xi_{n}}
\left[\frac{u_{+,1}(z)}{s_{11}(z)}\right]&=\frac{2b_{-}(\xi_{n})}{s_{11}''(\xi_{n})}
e^{2i\theta(\xi_{n})}\left[u'_{-,2}(\xi_{n})+(B_{-}
[\xi_{n}]+2i\theta'(\xi_{n}))u_{-,2}(\xi_{n}))\right],\tag{\theequation c}\\
\mathop{Res}_{z=\hat{\xi}_{n}}M_{2}^{+}(z)=\mathop{Res}_{z=\hat{\xi}_{n}}
\left[\frac{u_{+,2}(z)}{s_{22}(z)}\right]&=\frac{2b_{+}(\hat{\xi}_{n})}
{s_{22}''(\hat{\xi}_{n})}e^{-2i\theta(\hat{\xi}_{n})}\left[u'_{-,1}(\hat{\xi}_{n})+(B_{+}
[\hat{\xi}_{n}]-2i\theta'(\hat{\xi}_{n}))u_{-,1}(\hat{\xi}_{n}))\right].
\tag{\theequation d} \label{Residue4}
\end{align}
with the notion
\begin{align}
A_{+}[\hat{\xi}_{n}]=\frac{2b_{+}(\hat{\xi}_{n})}{s''_{22}(\hat{\xi}_{n})},\quad
\chi_{n}^{+}=B_{+}[\hat{\xi}_{n}]-2i\theta'(\hat{\xi}_{n}),\\
A_{-}[\xi_{n}]=\frac{2b_{-}(\xi_{n})}{s''_{11}(\xi_{n})},\quad
\chi_{n}^{-}=B_{-}[\xi_{n}]+2i\theta'(\xi_{n}).
\end{align}
\subsection{Formulation of the RH problem with double zeros}
In the analysis of simple zero, we have established the RH problem, which is still applicable to the double zeros, including the asymptotic behavior of matrix $M^{\pm}$ and jump matrix of $G$. To deal with the RH problem, we will regularize it via subtracting out the asymptotic behavior and the pole contributions, that is
\begin{align}\label{T111}
\begin{split}
 M^{+}-&\mathbb{I}+\frac{i}{z}\sigma_{3}Q_{-}-\sum_{n=1}^{2N}\left\{
\frac{\mathop{Res}\limits_{z=\hat{\xi}_{n}}M^{+}}{z-\hat{\xi}_{n}}
+\frac{\mathop{P_{-2}}\limits_{z=\hat{\xi}_{n}}M^{+}}{(z-\hat{\xi}_{n})^{2}}
+\frac{\mathop{Res}\limits_{z=\xi_{n}}M^{-}}{z-\xi_{n}}
+\frac{\mathop{P_{-2}}\limits_{z=\xi_{n}}M^{-}}{(z-\xi_{n})^{2}}\right\}\\&=
M^{-}-\mathbb{I}+\frac{i}{z}\sigma_{3}Q_{-}-\sum_{n=1}^{2N}\left\{
\frac{\mathop{Res}\limits_{z=\hat{\xi}_{n}}M^{+}}{z-\hat{\xi}_{n}}
+\frac{\mathop{P_{-2}}\limits_{z=\hat{\xi}_{n}}M^{+}}{(z-\hat{\xi}_{n})^{2}}
+\frac{\mathop{Res}\limits_{z=\xi_{n}}M^{-}}{z-\xi_{n}}
+\frac{\mathop{P_{-2}}\limits_{z=\xi_{n}}M^{-}}{(z-\xi_{n})^{2}}\right\}
-M^{-}G,
\end{split}
\end{align}
here the first seven terms on the right and left sides of the equation \eqref{T111} are analytic in $D^{-}$ and $D^{+}$, respectively. Meanwhile, the asymptotic behavior of the equation \eqref{T111} is  $O(1/z)$ $(z\rightarrow\infty)$ and $O(1)$ $(z\rightarrow0)$. As well as the asymptotic behavior of $G(z)$ is  clear, i.e., $O(1/z)$ $(z\rightarrow\infty)$ and $O(1)$ $(z\rightarrow0)$ along with the real axis.  Similar to simple  pole case, one can infer that using Cauchy operator and Plemelj's formulae
\begin{align}\label{T11}
\begin{split}
M(x,t;z)=&\mathbb{I}-\frac{i}{z}\sigma_{3}Q_{-}+\sum_{n=1}^{2N}\left\{
\frac{\mathop{Res}\limits_{z=\hat{\xi}_{n}}M^{+}}{z-\hat{\xi}_{n}}
+\frac{\mathop{P_{-2}}\limits_{z=\hat{\xi}_{n}}M^{+}}{(z-\hat{\xi}_{n})^{2}}
+\frac{\mathop{Res}\limits_{z=\xi_{n}}M^{-}}{z-\xi_{n}}
+\frac{\mathop{P_{-2}}\limits_{z=\xi_{n}}M^{-}}{(z-\xi_{n})^{2}}\right\} \\
&+\frac{1}{2\pi i}\int_{\Sigma}\frac{M^{-}(x,t;\zeta)G(x,t;\zeta)}{\zeta-z}\,d\zeta,
\quad z\in\mathbb{C}\setminus\Sigma.
\end{split}
\end{align}

To give the expression of $M$, a closed algebraic system is needed, so we will expand the following work according to this purpose. Firstly, the expression \eqref{Matr} implies at the zero conditions
\begin{align}\label{T12}
\begin{split}
\mathop{Res}_{z=\xi_{n}}[M^{-}]=\left(\mathop{Res}_{z=\xi_{n}}
\left[\frac{\mu_{+,1}(x,t;z)}{s_{11}(z)}\right],0\right), \quad \mathop{P_{-2}}_{z=\xi_{n}}M^{-}=\left(\mathop{P_{-2}}_{z=\xi_{n}}
\left[\frac{\mu_{+,1}(x,t;z)}{s_{11}(z)}\right],0\right),  \\
\mathop{Res}_{z=\hat{\xi}_{n}}[M^{+}]=\left(0,\mathop{Res}_{z=\hat{\xi}_{n}}
\left[\frac{\mu_{+,2}(x,t;z)}{s_{22}(z)}\right]\right), \quad \mathop{P_{-2}}_{z=\hat{\xi}_{n}}M^{+}=\left(0,\mathop{P_{-2}}_{z=\hat{\xi}_{n}}
\left[\frac{\mu_{+,2}(x,t;z)}{s_{22}(z)}\right]\right).
\end{split}
\end{align}
Combining with \eqref{Residue1}-\eqref{Residue4}
\begin{align}\label{T13}
\begin{split}
&\frac{\mathop{Res}\limits_{z=\hat{\xi}_{n}}M^{+}}{z-\hat{\xi}_{n}}
+\frac{\mathop{P_{-2}}\limits_{z=\hat{\xi}_{n}}M^{+}}{(z-\hat{\xi}_{n})^{2}}
+\frac{\mathop{Res}\limits_{z=\xi_{n}}M^{-}}{z-\xi_{n}}
+\frac{\mathop{P_{-2}}\limits_{z=\xi_{n}}M^{-}}{(z-\xi_{n})^{2}}=\\&
\left(H_{n}^{-}(z)\left[u'_{-,2}(\xi_{n})+
\left(\chi_{n}^{-}+\frac{1}{z-\xi_{n}}\right)u_{-,2}(\xi_{n})\right],
H_{n}^{+}(z)\left[u'_{-,1}(\hat{\xi}_{n})+\left(\chi_{n}^{+}+
\frac{1}{z-\hat{\xi}_{n}}\right)u_{-,1}(\hat{\xi}_{n})\right]\right),
\end{split}
\end{align}
with $H_{n}^{-}(z)=\frac{A_{-}[\xi_{n}]}{z-\xi_{n}}e^{2i\theta(\xi_{n})}$ and
$H_{n}^{+}(z)=\frac{A_{+}[\hat{\xi}_{n}]}{z-\hat{\xi}_{n}}
e^{-2i\theta(\hat{\xi}_{n})}$.
To get the expression of $u'_{-,2}(\xi_{n})$ $u'_{-,1}(\hat{\xi}_{n})$ $u_{-,1}(\hat{\xi}_{n})$ and $u_{-,2}(\xi_{n})$ in \eqref{T13}, we consider the second column of $M(z)$ defined by \eqref{T11} along with \eqref{T13} as $z=\xi_{s}$ ($s=1, 2, \cdots, 2N$), that is
\begin{align}\label{T14}
\begin{split}
u_{-,2}(x,t;\xi_{s})=&\left(\begin{array}{cc}
                       -\frac{iq_{-}}{\xi_{s}} \\
                        1
                     \end{array}\right)
+\frac{1}{2\pi i}\int_{\Sigma}\frac{(M^{-}G)_{2}(\zeta)}{\zeta-\xi_{k}}\,d\zeta \\ &+\sum_{n=1}^{2N}H_{n}^{+}(\xi_{s})
\left[u'_{-,1}(x,t;\hat{\xi}_{n})+\left(\chi_{n}^{+}+
\frac{1}{\xi_{s}-\hat{\xi}_{n}}\right)
u_{-,1}(x,t;\hat{\xi}_{n})\right].
\end{split}
\end{align}
Furthermore
\begin{align}\label{T15}
\begin{split}
u'_{-,2}(x,t;\xi_{s})=&\left(\begin{array}{cc}
                       \frac{iq_{-}}{\xi_{s}^{2}} \\
                        0
                     \end{array}\right)
+\frac{1}{2\pi i}\int_{\Sigma}\frac{(M^{-}G)_{2}(\zeta)}{(\zeta-\xi_{s})^{2}}\,
d\zeta \\
&-\sum_{n=1}^{2N}\frac{H_{n}^{+}(\xi_{s})}{\xi_{s}-\hat{\xi}_{n}}
\left[u'_{-,1}(x,t;\hat{\xi}_{n})+\left(\chi_{n}^{+}+
\frac{2}{\xi_{s}-\hat{\xi}_{n}}\right)
u_{-,1}(x,t;\hat{\xi}_{n})\right].
\end{split}
\end{align}
From the \eqref{Q25} and \eqref{Q40}, we arrive at
\begin{align}\label{T16}
\left\{\begin{aligned}
u_{-,2}(\xi_{s})&=-\frac{iq_{-}}{\xi_{s}}u_{-,1}(\hat{\xi}_{s}),\\
u'_{-,2}(\xi_{s})&=\frac{iq_{-}}{\xi_{s}^{2}}u_{-,1}(\hat{\xi}_{s})-
\frac{iq_{0}^{2}q_{-}}{\xi_{s}^{3}}u'_{-,1}(\hat{\xi}_{s}),
\end{aligned}\right.
\end{align}
substituting \eqref{T16} into \eqref{T14} and \eqref{T15}, then
\begin{align}\label{T17}
\left\{\begin{aligned}
&\sum_{n=1}^{2N}\left(H^{+}_{n}(\xi_{s})u'_{-,1}(x,t;\hat{\xi}_{n})+
\left[H^{+}_{n}(\xi_{s})\left(\chi^{+}_{n}+\frac{1}{\xi_{s}-\hat{\xi}_{n}}\right)+
\frac{iq_{-}}{\xi_{s}}\delta_{sn}\right]u_{-,1}(\hat{\xi}_{n})\right)\\
&+\left(\begin{array}{cc}
        -\frac{iq_{-}}{\xi_{s}} \\
         1
       \end{array}\right)
+\frac{1}{2\pi i}\int_{\Sigma}\frac{(M^{-}G)_{2}(\zeta)}{\zeta-\xi_{s}}\,d\zeta=0,\\
&\sum_{n=1}^{2N}\frac{H_{n}^{+}(\xi_{k})}{\xi_{s}-\hat{\xi}_{n}}\left(\left[
\left(\chi_{n}^{+}+\frac{2}{\xi_{s}-\hat{\xi}_{n}}\right)
+\frac{iq_{-}}{\xi_{k}^{2}}\delta_{sn}\right]u_{-,1}(\hat{\xi}_{n})-
\frac{iq_{-}q_{0}^{2}}{\xi_{s}^{3}}\delta_{sn}u'_{-,1}(\hat{\xi}_{n})\right) \\
&-\left(\begin{array}{cc}
        \frac{iq_{-}}{\xi_{s}^{2}} \\
         0
       \end{array}\right)
-\frac{1}{2\pi i}\int_{\Sigma}\frac{(M^{-}G)_{2}(\zeta)}{(\zeta-\xi_{s})^{2}}\,d\zeta=0.
\end{aligned}\right.
\end{align}
The equation \eqref{T17} generates an algebraic system  that includes the functions we need, that is, $u_{-,1}(x,t;\hat{\xi}_{n})$ and $u'_{-,1}(x,t;\hat{\xi}_{n})$. In additional, based on the obtained expressions, $u_{-,2}(x,t;\xi_{n})$ and $u'_{-,2}(x,t;\xi_{n})$ can be derived by \eqref{T16}. Finally,  the $M(z)$ can be expressed from \eqref{T11}.
\begin{prop}
The solution to the focusing Kundu-Eckhaus equation  with double poles is determined  by
\begin{align}\label{T18}
\begin{split}
q(x,t)=-q_{-}-&i\sum_{n=1}^{2N}A_{+}[\hat{\xi}_{n}]e^{-2i\theta(x,t;\hat{\xi}_{n})}
[u'_{-,11}(x,t;\hat{\xi}_{n})\\&+\chi_{n}^{+}u_{-,11}(x,t;\hat{\xi}_{n})]
+\int_{\Sigma}\frac{(M^{-}G)_{12}(x,t;\zeta)}{2\pi}\,d\zeta.
\end{split}
\end{align}
\end{prop}
For convenience, let¡¯s make the following notions:
$\tilde{M}=\left(\begin{array}{cc}
                  \tilde{M}^{11} & \tilde{M}^{12} \\
                  \tilde{M}^{21} & \tilde{M}^{22} \\
                \end{array}
              \right)$  and  $\tilde{M}^{(sj)}=(\tilde{m}_{kn}^{(sj)})_{(2N)\times(2N)}$ with
 $\tilde{m}_{kn}^{(11)}=H_{n}^{+}(\xi_{k})\left(\chi_{n}^{+}+\frac{1}{\xi_{k}-
 \hat{\xi}_{n}}\right)+\frac{iq_{-}}{\xi_{k}}\delta_{k,n}$, $\tilde{m}_{kn}^{(12)}=H_{n}^{+}(\xi_{k})$,
 $\tilde{m}_{kn}^{(21)}=\frac{H_{n}^{+}(\xi_{k})}{\xi_{k}-\hat{\xi}_{n}}\left(\chi_{n}^{+}
 +\frac{2}{\xi_{k}-\hat{\xi}_{n}}\right)+\frac{iq_{-}}
 {\xi_{k}^{2}}\delta_{k,n}$,
 $\tilde{m}_{kn}^{(22)}=\frac{H_{n}^{+}(\xi_{k})}{\xi_{k}-\hat{\xi}_{n}}-
 \frac{iq_{-}q_{0}^{2}}{\xi_{k}^{3}}\delta_{k,n}$;
$\omega_{n}^{(1)}=A_{+}[\hat{\xi_{n}}]e^{-2i\theta(\hat{\xi_{n}})}\chi_{n}^{+}$,
$\omega_{n}^{(2)}=A_{+}[\hat{\xi_{n}}]e^{-2i\theta(\hat{\xi_{n}})}$,
$\nu_{n}^{(1)}=\frac{iq_{-}}{\xi_{n}}$, $\nu_{n}^{(2)}=\frac{iq_{-}}{\xi_{n}^{2}}$,
$y=(y_{1}^{(1)}, \cdots, y_{2N}^{(1)}, y_{1}^{(2)}, \cdots, y_{2N}^{(2)})^{T}$ with
$y_{n}^{(1)}=u_{-,11}(\hat{\xi}_{n})$ and $y_{n}^{(2)}=u'_{-,11}(\hat{\xi}_{n})$ .
\begin{prop}
Next, we consider the solution without the reflection, that is, $G=0$, which leads to the integral part of the solution \eqref{T18} is equal to zero, then the solution can be written as
\begin{align}\label{T19}
q(x,t)=-q_{-}-i\sum_{n=1}^{2N}A_{+}[\hat{\xi}_{n}]e^{-2i\theta(\hat{\xi}_{n})}
[u'_{-,11}(\hat{\xi}_{n})+\chi_{n}^{+}u_{-,11}(\hat{\xi}_{n})],
\end{align}
which is equivalent to
\begin{align}\label{T20}
q(x,t)=-q_{-}+i\frac{\det\left(
                \begin{array}{cc}
                  \tilde{M} & \nu \\
                  \omega^{T} & 0 \\
                \end{array}
              \right)}{\det \tilde{M}}.
\end{align}
\end{prop}
\begin{proof}
Under the condition with reflection-less, $u_{-,11}(\hat{\xi}_{n})$ and $u'_{-,11}(\hat{\xi}_{n})$ are determined by
\begin{align}\label{T21}
\left\{\begin{aligned}
&\sum_{n=1}^{2N}\left(H^{+}_{n}(\xi_{s})u'_{-,11}(\hat{\xi}_{n})+
\left[H^{+}_{n}(\xi_{s})\left(\chi^{+}_{n}+\frac{1}{\xi_{s}-\hat{\xi}_{n}}\right)+
\frac{iq_{-}}{\xi_{s}}\delta_{sn}\right]u_{-,11}(\hat{\xi}_{n})\right)
=\frac{iq_{-}}{\xi_{s}}\\
&\sum_{n=1}^{2N}\frac{H_{n}^{+}(\xi_{k})}{\xi_{s}-\hat{\xi}_{n}}\left(\left[
\left(\chi_{n}^{+}+\frac{2}{\xi_{s}-\hat{\xi}_{n}}\right)
+\frac{iq_{-}}{\xi_{k}^{2}}\delta_{sn}\right]u_{-,1}(\hat{\xi}_{n})-
\frac{iq_{-}q_{0}^{2}}{\xi_{s}^{3}}\delta_{sn}u'_{-,1}(\hat{\xi}_{n})\right)
=\frac{iq_{-}}{\xi_{s}^{2}},
\end{aligned}\right.
\end{align}
which can be transformed into the following form by using the above notion
\begin{align}\label{T22}
Hy=\nu,
\end{align}
the equation \eqref{T22} can be solved by the Cramer's rule. Then substituting the obtained results into the first expression of \eqref{T19}, one has the another equation, that's means the proposition is proved.
\end{proof}
\subsection{Trace formulae and theta condition}
As usual, the trace $s_{11}(z)$ and $s_{22}(z)$ can be similarly expressed by scattering and reflection coefficients, the theta condition will further be given via the obtained trace formulae. From the analysis of the simple zeros case, we know that $s_{11}(z)$ and $s_{22}(z)$ are analytic in $D^{-}$ and $D^{+}$, respectively. Assuming they have double zeros in discrete spectrum $\mathbb{Z}$, one has
\begin{align}\label{T23}
\vartheta^{-}(z)=s_{11}(z)\prod_{n=1}^{2N}\frac{(z-z_{n}^{*})^{2}(z+q_{0}^{2}/z_{n})^{2}}
{(z-z_{n})^{2}(z+q_{0}^{2}/z_{n}^{*})^{2}},\\
\vartheta^{+}(z)=s_{22}(z)\prod_{n=1}^{2N}\frac{(z-z_{n})^{2}(z+q_{0}^{2}/z_{n}^{*})^{2}}
{(z-z_{n}^{*})^{2}(z+q_{0}^{2}/z_{n})^{2}},
\end{align}
which are of the same asymptotic  behavior as $s_{11}(z)$ and $s_{22}(z)$ as $z\rightarrow\infty$ and have no zeros  point in their respective regions $D^{\pm}$.
 Note that $\vartheta^{+}\vartheta^{-}=s_{11}(z)s_{22}(z)$ and $\det S(z)=s_{11}(z)s_{22}(z)-s_{12}(z)s_{21}(z)=1$, which lead to
\begin{align}\label{T24}
\frac{1}{s_{11}(z)s_{22}(z)}=1-\frac{s_{21}(z)}{s_{11}(z)}\frac{s_{12}(z)}{s_{22}(z)}
=1-\rho(z)\tilde{\rho}(z)=1+\rho(z)\rho^{*}(z^{*}).
\end{align}
Obviously
\begin{align}\label{T25}
\vartheta^{+}(z)\vartheta^{-}(z)=\frac{1}{1+\rho(z)\rho^{*}(z^{*})}.
\end{align}
Taking the logarithm on both sides of the equation \eqref{T25}, we can get a scalar RH problem similar to the case of a single pole. By solving it with a similar method, the trace formula can be given as
\begin{align}
s_{11}(z)&=exp\left(-\frac{1}{2\pi i }\int_{\Sigma}\frac{\log[1+\rho(\zeta)\rho^{*}(\zeta^{*})]}{\zeta-z}
\,d\zeta\right)\prod_{n=1}^{2N}\frac{(z-z_{n})^{2}(z+q_{0}^{2}/z_{n}^{*})^{2}}
{(z-z_{n}^{*})^{2}(z+q_{0}^{2}/z_{n})^{2}},\\
s_{22}(z)&=exp\left(\frac{1}{2\pi i }\int_{\Sigma}\frac{\log[1+\rho(\zeta)\rho^{*}(\zeta^{*})]}{\zeta-z}
\,d\zeta\right)\prod_{n=1}^{2N}\frac{(z-z_{n}^{*})^{2}(z+q_{0}^{2}/z_{n})^{2}}
{(z-z_{n})^{2}(z+q_{0}^{2}/z_{n}^{*})^{2}}.
\end{align}
In additional, the theta condition can be derived via the obtained trace formulae as $z\rightarrow0$
\begin{align}
\arg\frac{q_{-}}{q_{+}}=argq_{-}-argq_{+}=\frac{1}{2\pi}\int_{\Sigma}
\frac{\log[1+\rho(\zeta)\rho^{*}(\zeta^{*})]}{\zeta}\,d\zeta+8\sum_{n=1}^{N}\arg z_{n}.
\end{align}
Considering the reflection-less, then the theta condition will be
\begin{align}
\arg\frac{q_{-}}{q_{+}}=argq_{-}-argq_{+}=8\sum_{n=1}^{N}\arg z_{n}.
\end{align}

In order to further observe the behavior of the solution \eqref{T20}, selecting the following parameters,  we get the following four figures through Maple

{\rotatebox{0}{\includegraphics[width=2.75cm,height=2.5cm,angle=0]{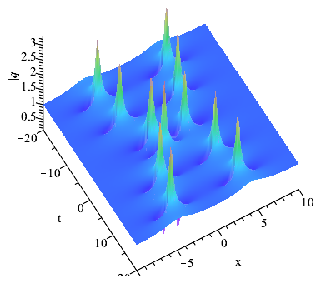}}}
{\rotatebox{0}{\includegraphics[width=2.75cm,height=2.5cm,angle=0]{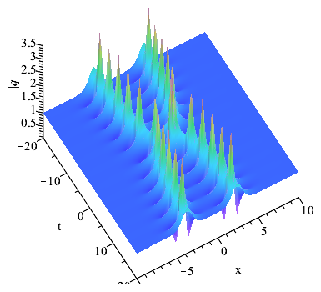}}}
{\rotatebox{0}{\includegraphics[width=2.75cm,height=2.5cm,angle=0]{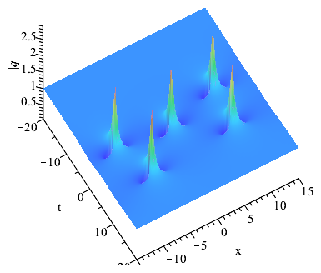}}}
{\rotatebox{0}{\includegraphics[width=2.75cm,height=2.5cm,angle=0]{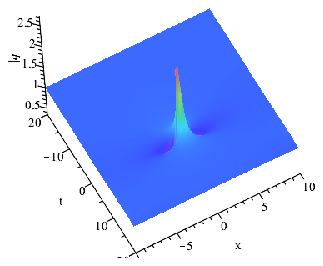}}}

 $\qquad\qquad(\textbf{a})\qquad \ \qquad\qquad\qquad(\textbf{b})
\ \qquad\quad\quad\qquad\qquad(\textbf{c})\qquad\qquad\qquad\qquad(\textbf{d})
\ \qquad$\\
\noindent { \small \textbf{Figure 7.} The soliton solutions  with the fixed parameters $N=1$,  $A_{-}[z_{1}]=B_{-}[z_{1}]=1$ and $q_{-}=1$. $\textbf{(a)}$: the breather solution with the $z_{1}=3i/2$; $\textbf{(b)}$: the breather solution with the $z_{1}=i/2$; $\textbf{(c)}$: the breather solution with the $z_{1}=1/4+i$; $\textbf{(d)}$: the rouge wave with the $z_{1}=1/8+i$.}

With the decrease of discrete eigenvalues, the number of breather solutions increases, which can be obtained naturally by comparing Fig. $7(a)$ and Fig. $7(b)$. It should be noted that not the smaller the eigenvalue, the more the number of breather solutions, but the solutions become irregular. As  It is interesting that when the eigenvalue $z_{1}=i$, it is the singular point of the solution \eqref{T20} for the KE equation, however when the eigenvalue is added to the real part, the solution tends to be the rouge wave solution as the real part tends to zero as shown in Fig. $7(d)$.

\section{Conclusions and discussions}

This work gives a detailed study of focusing Ke equation with non-zero boundary value at infinity, and gives a systematic answer to the questions in the introduction. Also similar to Biondini's work, an affine transformation is introduced to overcome the multi-valued function in order to construct the RH problem in the inverse scattering process. The next structure frame is to use the analytical and symmetric properties of Jost function and scattering coefficients to get the corresponding set of discrete spectrum points, and also the corresponding residue conditions. Based on these results, the form solution of focusing KE equation is obtained by solving the RH problem established in the inverse scattering process, and the propagation of the focusing KE equation solution under the non-reflection conditions is presented by using the software Maple. It should be noted that for most nonlinear partial differential equations, the solitons are related to the zero point of the analytical scattering coefficients. Therefore we further discuss the case when there are doubles zeros in the scattering coefficients, and obtain the doubles zeros soliton solutions, which is similarly the case simple zeros, but also has its differences.
\section*{Acknowledgements}
 
This work was supported by    the Postgraduate Research and Practice of Educational Reform for Graduate students in CUMT under Grant No. 2019YJSJG046, the Natural Science Foundation of Jiangsu Province under Grant No. BK20181351, the Six Talent Peaks Project in Jiangsu Province under Grant No. JY-059, the Qinglan Project of Jiangsu Province of China, the National Natural Science Foundation of China under Grant No. 11975306, the Fundamental Research Fund for the Central Universities under the Grant Nos. 2019ZDPY07 and 2019QNA35, and the General Financial Grant from the China Postdoctoral Science Foundation under Grant Nos. 2015M570498 and 2017T100413.

\end{document}